\documentclass[11pt]{article}
\usepackage{etoolbox}

\usepackage[margin=0.96in]{geometry}

\usepackage{amsmath}
\usepackage{amsthm}
\usepackage{amssymb}
\usepackage{enumerate}
\usepackage{hyperref}
\usepackage{authblk}
\newtheorem{theorem}{Theorem}[section]

\newtheorem{corollary}[theorem]{Corollary}
\newtheorem{proposition}[theorem]{Proposition}
\newtheorem{lemma}[theorem]{Lemma}
\newtheorem{claim}[theorem]{Claim}

\newcommand{\R}{\mathbb{R}}
\newcommand{\Z}{\mathbb{Z}}

\newcommand{\ind}[1]{\ensuremath{d_{#1}}}
\newcommand{\ketd}[1]{\ensuremath{4(d_{#1}+2)}}
\newcommand{\negd}[1]{\ensuremath{(d_{#1}+1)}}
\newcommand{\egyd}[1]{\ensuremath{(d_{#1}+2)}}
\newcommand{\itr}[1]{^{(#1)}}
\newcommand{\estep}{\textsc{Elementary step}}
\newcommand{\threshold}{\ensuremath{1/{(17m\bar B^{3})}}}
\newcommand{\Deltabound}{\ensuremath{n\bar B^{2}}}
\newcommand{\ole}{\overleftrightarrow}
\newcommand{\obe}{\overleftarrow}
\newcommand{\tf}{\ensuremath{{\tilde f}}}

\newcommand{\arcbound}{\ensuremath{17m\Delta}}
\newcommand{\sink}{\ensuremath{t}}

\newcommand{\lengthtotal}{\ensuremath{390n^3m}}
\newcommand{\ka}{\ensuremath{k}}

\title{A strongly polynomial algorithm for generalized flow maximization}
\author{L\'aszl\'o A. V\'egh
\\
\emph{Department of Mathematics}\\
\emph{London School of Economics and Political Science}\\
\texttt{L.Vegh@lse.ac.uk}
}
\date{}

\begin{document}
\maketitle
\begin{abstract}
A strongly polynomial algorithm is given for the generalized flow
maximization problem. It uses a  new variant of the scaling technique,
called continuous scaling.
The main measure of progress is that within a strongly polynomial number of steps, an arc can be identified 
that must be tight in every dual optimal solution, and thus can be
contracted. As a consequence of the result, we also obtain a strongly polynomial algorithm for the linear feasibility problem with at most two nonzero entries per column in the constraint matrix.
\end{abstract}

\section{Introduction}
The generalized flow model is a classical extension of network flows. Besides the capacity constraints, for every arc $e$ there is a gain factor $\gamma_e>0$, such that flow amount gets multiplied by $\gamma_e$ while traversing  arc $e$. We study the flow maximization problem, where the objective is to send the maximum amount of flow to a sink node $\sink$.
The model was already formulated by Kantorovich \cite{Kantorovich}, as
one of the first examples of Linear Programming; it has several
applications in Operations Research \cite[Chapter 15]{amo}. Gain
factors can be used to model physical changes
such as leakage or theft. 
Other common 
applications use the nodes to represent different types of
entities, e.g. different currencies,
and the gain factors correspond to the exchange rates.

The existence of a strongly polynomial algorithm for Linear Programming is a major open question in the theory of computation. 
 This refers to an algorithm with
 the number of arithmetic operations  polynomially bounded in the number  of variables and constraints, and the size of the numbers during the computations polynomially bounded in the input size.
 A landmark result by Tardos \cite{Tardos86} is an algorithm with
 the running time dependent only on the size of numbers in the
 constraint matrix, but independent from the right-hand side and the
 objective vector. This gives strongly polynomial algorithms for
 several combinatorial problems such as minimum cost flows (see also
 Tardos \cite{Tardos85strong}) and multicommodity flows.

Instead of bounding the sizes of numbers, one might impose structural restrictions on the
constraint matrix. A natural question arises whether there
exists a  strongly polynomial algorithm for linear programs (LPs)
with at most two nonzero entries per column (that can be arbitrary numbers).
This question is still open; as shown by Hochbaum \cite{hochbaum04}, all such LPs can be 
efficiently transformed to equivalent instances of the minimum cost generalized flow problem.
(Note also that every LP can be efficiently transformed to an equivalent one
with at most three nonzero entries per column.)
In 1983, Megiddo \cite{Megiddo83} gave a strongly polynomial algorithm for solving the dual feasibility problem for such LPs; he introduced the concept of strongly polynomial algorithms in the same paper.
A corollary of our result is the first strongly polynomial algorithm for the primal feasibility problem.

Generalized flow maximization is probably the simplest natural class of LPs where no strongly polynomial algorithm was known.
The existence of such an algorithm has been a well-studied and longstanding open problem (see e.g. \cite{Goldberg91,Cohen94,Wayne02,Radzik04,Shigeno07}).
 A strongly polynomial algorithm for a restricted class was given by Adler and Cosares \cite{Adler91}.

\medskip

In this paper, we exhibit a strongly polynomial algorithm for generalized flow maximization.
Let $n$ denote the number of nodes and $m$ the number of arcs in the network, and let $B$ denote the largest integer used in the description of the input (see Section~\ref{sec:probdef} for the precise problem definition).
A strongly polynomial algorithm for the problem entails the following (see \cite{gls}):
 {\em(i)} it uses only elementary arithmetic operations (addition, subtraction, multiplication, division), and comparisons; {\em (ii)} the
number of these operations is bounded by a polynomial of $n$ and $m$; {\em
  (iii)} all numbers occurring in
the computations are rational numbers of encoding size polynomially bounded in 
$n$, $m$ and $\log B$ -- or equivalently, it is a polynomial space algorithm.
Here,  the encoding size of a positive rational number $p/q$ is defined as
$\lceil \log_2 (p+1)\rceil+\lceil \log_2 (q+1)\rceil$. By the running time of a strongly polynomial algorithm we mean the total number of elementary arithmetic operations and comparisons.

\medskip
Combinatorial approaches have been  applied to generalized flows already in the sixties by  Dantzig \cite{Dantzig63} and Jewell \cite{Jewell62}.
However, the first polynomial-time combinatorial algorithm was only given in 1991 by Goldberg, Plotkin and Tardos \cite{Goldberg91}. This was followed by a multitude of further combinatorial algorithms e.g.\ \cite{Cohen94,Goldfarb96,Goldfarb97,Tardos98,Fleischer02,Goldfarb02a,Goldfarb02,Wayne02,Radzik04,Restrepo09,Vegh11}; a central motivation of this line of research was to develop a strongly polynomial algorithm. 
The algorithms of Cohen and Megiddo \cite{Cohen94}, Wayne \cite{Wayne02}, and Restrepo and Williamson \cite{Restrepo09}  present fully polynomial time approximation schemes, that is,  for every $\varepsilon>0$, they can find a solution within   $\varepsilon$ from the optimum value in running time polynomial in $n$, $m$ and $\log(1/\varepsilon)$. This can be transformed to an optimal solution for a sufficiently small $\varepsilon$; however, this value does depend on $B$ and hence the overall running time will also depend on $\log B$.
The current most efficient weakly polynomial algorithms are the interior point approach of Kapoor an Vaidya \cite{Kapoor96} with running time 
$O(m^{1.5}n^2\log B)$, and  the  combinatorial algorithm
 by
Radzik \cite{Radzik04} with running time $\tilde O(m^2n\log B)$.\footnote{The $\tilde O()$ notation hides a polylogarithmic factor.} For a survey on combinatorial generalized flow algorithms,   see Shigeno \cite{Shigeno07}.

\medskip

The generalized flow maximization problem exhibits deep structural similarities to the minimum cost circulation problem, as first pointed out by Truemper \cite{Truemper77}. 
Most combinatorial algorithms for generalized flows, including both algorithms by Goldberg et al. \cite{Goldberg91}, exploit this analogy and adapt existing efficient techniques from minimum cost circulations.
For the latter problem, several strongly polynomial algorithms are known, the first one given by Tardos \cite{Tardos85strong}; others relevant to our discussion are those by Goldberg and Tarjan \cite{Goldberg89}, and by Orlin \cite{Orlin93}; see also \cite[Chapters 9-11]{amo}.
Whereas these algorithms serve as starting points for most generalized flow algorithms, the applicability of the techniques is by no means obvious, and different methods have to be combined.
As a consequence, the strongly polynomial analysis cannot be carried over when adapting minimum cost circulation approaches to generalized flows,  although weakly polynomial bounds can be shown. To achieve a strongly polynomial guarantee, further new algorithmic ideas are required that are specific to the structure of generalized flows. The new ingredients of our algorithm are highlighted in Section~\ref{sec:overview}.

Let us now outline the scaling method for minimum cost circulations, a motivation of our generalized flow algorithm.
The first (weakly) polynomial time algorithm for minimum cost circulations was given by Edmonds and Karp \cite{Edmonds72}, introducing the simple yet powerful idea of scaling (see also \cite[Chapter 9.7]{amo}). The algorithm consists of $\Delta$-phases, with the value of $\Delta>0$ decreasing by a factor of at least two between every two phases, yielding an optimal solution for sufficiently small $\Delta$. In the $\Delta$-phase, the flow is transported in units of $\Delta$ from nodes with excess to nodes with deficiency using shortest paths in the graph of arcs with residual capacity at least $\Delta$. Orlin \cite{Orlin93}, (see also \cite[Chapters 10.6-7]{amo}) devised a strongly polynomial version of this algorithm. The key notion is that of ``abundant arcs''.
In the $\Delta$-phase of the scaling algorithm \cite{Edmonds72}, the arc $e$ is called {\em abundant} if it carries  $>4n\Delta$ units of flow. For such an arc $e$, it can be shown that $x^*_e>0$ must hold for some optimal solution $x^*$. By primal-dual slackness, the corresponding constraint must be tight in every dual optimal solution. Based on this observation, Orlin \cite{Orlin93} showed that such an arc can be contracted; the scaling algorithm is then restarted on the smaller graph. This leads to  a dual optimal solution in strongly polynomial time; that provided, a primal optimal solution can be found via a single maximum flow computation. Orlin \cite{Orlin93} also presents a more efficient but also more sophisticated  implementation of this idea. 

\medskip

Let us now turn to generalized flows. 
The analogue of the scaling method was an important component of the \textsc{Fat-Path} algorithm of \cite{Goldberg91}; the algorithm of Goldfarb, Jin and Orlin \cite{Goldfarb97} and the one in \cite{Vegh11} also use this technique. The notion of ``abundant arcs'' can be easily extended to these frameworks: if an arc $e$ carries a ``large'' amount of flow as compared to $\Delta$, then it must be tight in every dual optimal solution, and hence can be contracted. This idea was already used by Radzik \cite{Radzik04}, to boost the running time of \cite{Goldfarb97}. Nevertheless, it is not known whether an ``abundant arc'' would always appear in any of the above algorithms within a strongly polynomial number of steps.

Our contribution is a new type of scaling algorithm that suits better the dual structure of the generalized flow problem, and thereby the quick appearance of an ``abundant arc'' will be guaranteed.
 Whereas in all previous methods, the scaling factor $\Delta$ remains constant for a linear number of path augmentations, our {\em continuous scaling method} keeps it decreasing in every elementary iteration of the algorithm, even in those that lead to finding the next augmenting path. 

\medskip

The rest of the paper is structured as follows. Section~\ref{sec:probdef} first defines the problem setting,  introduces  relabelings,  gives the characterization of optimality, and defines the notion of $\Delta$-feasibility.
Section~\ref{sec:overview} then gives a more detailed account of the main algorithmic ideas. 

The algorithm is presented in three different versions. First, Section~\ref{sec:weak} describes a relatively simple scaling algorithm called \textsc{Continuous Scaling}, with a weakly polynomial running time guarantee proved in Section~\ref{sec:anal-weak}. Our strongly polynomial algorithm \textsc{Enhanced Continuous Scaling} in Section~\ref{sec:strong} builds on this, by including one additional subroutine, and a framework for contracting arcs. The running time analysis is given in Section~\ref{sec:strong-analysis}. This achieves a strongly polynomial  bound on the number of steps. A strongly polynomial algorith must also satisfy requirement {\em(iii)} on bounded number sizes. This requires further modifications of the algorithm in Section~\ref{sec:rounding} by introducing certain rounding steps.

 Section~\ref{sec:further}
shows reductions between different formulations; in particular, the corollary on LP feasibility problems with at most two nonzeros per column is shown here. Section~\ref{sec:further} is independent from the preceding sections and can be read directly after Section~\ref{sec:probdef}.
Section~\ref{sec:final} concludes with some additional remarks and open questions.

\section{Preliminaries}\label{sec:probdef}
We start by introducing the most general formulation our approach is applicable to. Consider the linear feasibility problem
\begin{align}
Ax&=b\tag{LP2}\label{lp2}\\
0\le x&\le u,\notag
\end{align}
such that every column of $A$ contains at most two nonzero entries. By making use of a reduction by Hochbaum \cite{hochbaum04}, in Section~\ref{sec:lp2} we show that every problem of this form can be reduced to the generalized flow maximization problem as defined next.

Let $G=(V,E)$ be a directed graph with a designated sink node $\sink\in V$.
Let $n=|V|$, $m=|E|$, and for each node $i\in V$, let $\ind{i}$ 
denote total number of arcs incident to $i$ (both entering and leaving).
We will always assume $n\le m$. We do not allow
parallel arcs and hence we may use $ij$ to denote the arc from $i$ to
$j$. This is for notational convenience only, and all results straightforwardly extend to the 
setting with parallel arcs. All paths and cycles in the paper will refer to directed paths and directed cycles.

The following is the {\em standard formulation} of the problem. Let us be
given arc capacities $u:E\rightarrow\mathbb{Q}_{>0}\cup \{\infty\}$ and
gain factors $\gamma:E\rightarrow \mathbb{Q}_{>0}$. 
\begin{align}
\max \sum_{j:j\sink\in E}\gamma_{j\sink}f_{j\sink}-\sum_{j:\sink j\in
  E}f_{\sink j}& \notag\\
\sum_{j:ji\in E}\gamma_{ji}f_{ji}-\sum_{j:ij\in E}f_{ij}&\ge 0\quad 
\forall i\in V-\sink{} \tag{$P_u$}\label{primal-standard}\\
0\le f&\le u\notag
\end{align}
It is common in the literature to define the problem with equalities in the node constraints.
The two forms are essentially equivalent, see e.g. \cite{Shigeno07}; moreover, the form with equality is often
solved via a reduction to (\ref{primal-standard}).
In this paper, we prefer to use yet another equivalent formulation,
where the arcs have no upper capacities,
but there are node demands instead.
A problem given in the standard
formulation can be easily transformed to an equivalent instance in
this form; the transformation is described in
Section~\ref{sec:transform}.
Given a node demand vector $b:V\rightarrow \mathbb{Q}$ and gain factors
$\gamma:E\rightarrow \mathbb{Q}_{>0}$, the {\em uncapacitated formulation} 
is defined as
\begin{align}
\max \sum_{j:j\sink\in E}\gamma_{j\sink}f_{j\sink}-\sum_{j:\sink j\in
  E}f_{\sink j}& \notag\\
\sum_{j:ji\in E}\gamma_{ji}f_{ji}-\sum_{j:ij\in E}f_{ij}&\ge b_i\quad
\forall i\in V-\sink{} \tag{$P$}\label{primal}\\
0\le f&\notag
\end{align}
Note the value of $b_t$ is irrelevant as it is not present in the
formulation; we may e.g. assume $b_t=0$.
For a  vector $f\in \R^{|E|}_{\ge 0}$, let us define the {\em excess} of a node $i\in V$ by
\[
e_i(f):=\sum_{j:ji\in E}\gamma_{ji}f_{ji}-\sum_{j:ij\in E}f_{ij}-b_i.
\]
The node constraints in (\ref{primal}) can be written as $e_i(f)\ge
0$, and the objective is equivalent to maximizing $e_\sink(f)$. When $f$ is clear from the context, we will denote the
excess simply by $e_i:=e_i(f)$. 
By a {\em generalized flow} we mean a feasible solution to (\ref{primal}),
that is, a nonnegative vector $f\in \R^{|E|}_{\ge 0}$ with
$e_i(f)\ge 0$ for all $i\in V- \sink$.
Let us define the {\em surplus} of $f$
as 
\[
Ex(f):=\sum_{i\in V-\sink{}}e_i(f).
\]

It will be convenient to make the following assumptions.
\begin{align}
&\mbox{There is an arc }{i\sink{}\in E\mbox{ for every }i\in
  V-\sink{}};
\tag{$\star$}\label{cond:root}\\
&\mbox{The problem (\ref{primal}) is feasible, and an initial feasible solution $\bar f$ is provided.}\tag{$\star\star$}\label{cond:init}\\
&\mbox{The objective value in (\ref{primal}) is bounded.}\tag{$\star\star\star$}\label{cond:bounded}
\end{align}
 These assumptions are without loss of generality; it is shown in Section~\ref{sec:transform} that any problem in the standard form can be transformed to an equivalent one in the uncapacitated form that also satisfies assumptions \eqref{cond:root} and \eqref{cond:init}. Condition \eqref{cond:root} can be easily achieved by adding new arcs to the sink with gain factors small enough not to influence the solution. To obtain \eqref{cond:init}, observe that $f\equiv 0$ is feasible to (\ref{primal-standard}); $\bar f$ in (\ref{cond:init}) will be the image of $0$ under
the transformation. To justify \eqref{cond:bounded}, in the same Section~\ref{sec:transform} we show how unboundedness can be detected.
Furthermore, in Section~\ref{sec:lp2} we show how an arbitrary
instance of \eqref{lp2} can be reduced to solving two instances of \eqref{primal} satisfying these assumptions.

Let us introduce some further notation. For an arc set $H\subseteq E$,
let $\obe H$ denote the set of reverse arcs, that is, $\obe H:=\{ji:
ij\in H\}$; let $\ole H:=H\cup \obe H$.
We define the gain factor of a reverse arc $ji\in \obe H$ by $\gamma_{ji}:=1/\gamma_{ij}$.
For an arc set $F\subseteq E$ and  node sets $S,T\subseteq V$, let
$F[S,T]:=\{ij\in F: i\in S, j\in T\}$. We also use $F[S]:=F[S,S]$ to
denote the set of arcs in $F$ spanned by $S$. For a node $i\in V$, let
$\delta^{in}(i)$ and $\delta^{out}(i)$ denote the set of arcs entering
and leaving $i$, respectively.  We will use the vector norms $||x||_1=\sum_i|x_i|$ and $||x||_\infty=\max_i|x_i|$.
For integers $a\le b$, let $[a,b]:=\{a,a+1,\ldots,b\}$.

A vector $f:\ole E\rightarrow \mathbb{R}_{\ge 0}$ is called a {\em path flow}, if its support is a path $P=w_1w_2\ldots w_t\subseteq \ole E$, 
and $\gamma_{w_\ell}f_{w_\ell}=f_{w_{\ell+1}}$ for every $1\le \ell\le t-1$.
In other words, the incoming flow equals the outgoing flow in every
internal node of the path. We say that a path flow $f$ sends $\alpha$
units of flow from $p$ to $q$, if the support of $f$ is a $p-q$ path, and the flow value arriving at $q$ 
 equals $\alpha$. Note however, that the amount of flow leaving $p$ is typically different from $\alpha$.

\subsection{Encoding size}\label{sec:encoding}
In the weakly polynomial algorithm, the running time will be dependent
on the encoding size of the input, that consists of rational
numbers. In a strongly polynomial algorithm, all numbers appearing during the computations must 
be rational of encoding size polynomially bounded in the input size.
(We remark that the notion of strongly polynomial algorithms is also applicable to problems with
arbitrary real numbers in the input; this model assumes that every basic arithmetic operation can be carried out in
$O(1)$ time.) 
\medskip

\noindent {\bf Standard formulation.}
We are given an integer $B$ such that all capacities $u$ and gain factors $\gamma$ are
rational numbers, given as quotients of two integers $\le B$.

\medskip

\noindent {\bf Uncapacitated formulation.} 
We give more complicated conditions
on the encoding size of the different quantities. This is in order to maintain good bounds on the encoding size when transforming an
instance from the standard  to the uncapacitated
formulation in Section~\ref{sec:transform}.

Assume the instance satisfies
conditions (\ref{cond:root}), (\ref{cond:init}) and (\ref{cond:bounded}).
We use the integer $\bar B$ to bound the encoding size of the input as follows.
\begin{itemize}
\item 
The arcs can be classified into two types, {\em regular} and {\em auxiliary}, with $\sink$ being the endpoint of every auxiliary arc.
For a regular arc $ij$, the gain factor $\gamma_{ij}$ is given as a rational number, such that $\bar B$ is an integer multiple of the product of the numerators and denominators of all $\gamma_{ij}$ values for regular arcs. For every auxiliary arc $i\sink$, $\gamma_{i\sink}
={1}/\bar B$.
\item 
For every $i\in V-\sink$, $|b_i|\le \bar B$, and is an integer multiple of $1/\bar B$. 
\item
For the initial solution $\bar f$, and for every  $ij\in E$,
$\bar f_{ij}\le \bar B$ and $\bar f_{ij}$ is an
integer multiple of $1/\bar B$.
\end{itemize}

\smallskip

The reduction in Section~\ref{sec:transform} will transform an
instance in the standard formulation with $n$ nodes and $m$ arcs and parameter $B$ to an
uncapacitated instance with $\le m+n$ nodes, $\le 2m$ arcs and $\bar B\le 2B^{4m}$.

\medskip

Our main result is the following.
\begin{theorem}\label{thm:main}
There exists an $O(n^3m^2)$ time strongly polynomial algorithm for the uncapacitated
formulation \eqref{primal} with assumptions (\ref{cond:root}),
(\ref{cond:init}) and (\ref{cond:bounded}). 
\end{theorem}
Using the transformation in  Section~\ref{sec:transform}, this gives
an $O(m^5)$ time strongly polynomial algorithm for the standard
formulation \eqref{primal-standard}. Finally, using the reduction in
Section~\ref{sec:lp2}, we get an $O(m^5)$ algorithm for the linear
feasibility problem \eqref{lp2} with $n$ constraints and $m$ variables.

\subsection{Labelings and optimality conditions}\label{sec:genflow-opt}
Dual solutions to (\ref{primal}) play a crucial role in the entire
generalized flow literature. Let $y:V\rightarrow \R_{\ge0}$ be a
solution to the dual of (\ref{primal}). 
Following Glover and Klingman \cite{Glover73}, the literature standard
is not to consider the $y$ values but their inverses instead.
With $\mu_i:=1/y_i$, we can write the dual of (\ref{primal}) in the following form.
\begin{align}
\max& \sum_{i\in V}\frac{b_i}{\mu_i}\notag\\
\gamma_{ij}{\mu_i}&\le \mu_j\quad \forall ij\in E\tag{$D$}\label{dual}\\
\mu_i&>0\quad \forall i\in V-\sink\notag\\
\mu_t&=1\notag
\end{align}
A feasible solution $\mu$ to this program will be called a {\em relabeling} or {\em labeling}. An {\em optimal labeling} is an optimal solution to (\ref{dual}).
Whereas there could be values $\mu_i=\infty$ corresponding to $y_i=0$, assumption (\ref{cond:root}) guarantees that
all $\mu_i$ values must be finite.
A useful and well-known property is the following. 
\begin{proposition}
Given an optimal solution to (\ref{dual}), an optimal solution to (\ref{primal}) can be obtained in strongly polynomial time, and conversely,
given an optimal solution to (\ref{primal}), an optimal solution to (\ref{dual}) can be obtained in strongly polynomial time.
\end{proposition}
In fact, our strongly polynomial algorithm proceeds via finding an optimal solution to (\ref{dual}), and computing the primal optimal solution via a single maximum flow computation. 
The first part of the above proposition is proved in
Theorem~\ref{thm:tight-flow}(i), whereas the second part (which is not
needed for our algorithm) can be shown using an argument similar to
the proof of Lemma~\ref{lem:init}.

Relabelings will be used in all parts of the algorithm and proofs. For
a generalized flow $f$ and a labeling $\mu$, we define the relabeled
flow $f^\mu$ by 
\[f^\mu_{ij}:=\frac{f_{ij}}{\mu_i}\]
for all $ij\in E$. This can be interpreted as changing the base unit of measure at the nodes 
(i.e. in the example of the currency exchange network, it corresponds to changing the unit from pounds to pennies).
 To get a problem setting equivalent to the original one, we have to relabel all other quantities accordingly. That is, we define
relabeled gains, demands, excesses and surplus by
\[ \gamma_{ij}^\mu:=\gamma_{ij}\frac{\mu_i}{\mu_j},\quad
 b^\mu_i:=\frac{b_i}{\mu_i},\quad e^\mu_i:=\frac{e_i}{\mu_i}, \mbox{ and }\quad Ex^\mu(f):=\sum_{i\in V-\sink}e_i^\mu,
\]
 respectively.
Another standard notion is the {\em residual network} $G_f=(V,E_f)$ of a generalized flow $f$, defined as
\[
E_f:=E\cup \{ij: ji\in E, f_{ji}>0\}.
\]
Arcs in $E$  are called \textsl{forward arcs}, while arcs in the second set are
\textsl{reverse arcs}. 
Recall that for a reverse arc $ji$ we defined $\gamma_{ji}=1/\gamma_{ij}$.
Also, we define $f_{ji}:=-\gamma_{ij}f_{ij}$ for every reverse arc $ji\in E_f$.
By increasing (decreasing) $f_{ji}$ by $\alpha$ on a reverse arc $ji\in E_f$, we mean decreasing (increasing) $f_{ij}$ by $\alpha/\gamma_{ij}$.
The input graph $G=(V,E)$ is allowed to have pairs of oppositely
directed arcs $ij$ and $ji$, making
our notation slightly ambiguous: for an arc $ij$, we will denote its
reverse arc by $ji$, which might be an arc parallel to the original
arc from $j$ to $i$ in the input. However, this should not be a source of confusion: whenever the arc $ji$ is mentioned in the context of $ij$, it will always refer to the reverse arc.

The crucial notion of conservative labelings is motivated by primal-dual slackness.
Let $f$ be a generalized flow (that is, a feasible solution to
(\ref{primal})), and let $\mu: V\rightarrow \R_{>0}$.
We say that $\mu$ is a {\em conservative labeling} for $f$, if
$\mu$ is a feasible solution to (\ref{dual}) with the further requirement that 
$\gamma^{\mu}_{ij}= 1$ whenever $f_{ij}>0$ for $ij\in E$.
The following characterization of optimality is a straightforward consequence of primal-dual slackness in Linear Programming.
We state the optimality conditions both for the uncapacitated
formulation (\ref{primal}), and for the standard formulation
(\ref{primal-standard}). In the latter part we do not assume
(\ref{cond:root}), and therefore $\mu_i=\infty$ is also allowed.

\begin{theorem}\label{thm:genflow-opt}
\begin{enumerate}[(i)]
\item
 Assume  (\ref{cond:root}) holds. 
A generalized flow $f$ is an optimal solution to (\ref{primal}) if and only if
there exists a finite conservative labeling $\mu$, and  $e_i=0$ for
all $i\in V-\sink$.
\item
A feasible solution $f$ to the standard form
(\ref{primal-standard}) is optimal if and only if there exists a
function $\mu: V\rightarrow \R_{>0}\cup\{\infty\}$
such that $\mu_\sink=1$, and 
$\gamma_{ij}\mu_i\le \mu_j$ if $f_{ij}=0$, $\gamma_{ij}\mu_i= \mu_j$ if
$0<f_{ij}<u_{ij}$, and $\gamma_{ij}\mu_i\ge \mu_j$ if $f_{ij}=u_{ij}$; 
further, $e_i=0$ whenever 
 $\mu_i<\infty$.
\end{enumerate}
\end{theorem}

Given a labeling $\mu$, we say that an arc $ij\in E_f$ is {\em tight}
if $\gamma_{ij}^\mu=1$. A directed path in $E_f$ is called {\em tight} if it
consists of tight arcs.

\subsection[Delta-feasible labels]{$\Delta$-feasible
  labels}\label{sec:cons-label}
Let us now introduce a relaxation of conservativity crucial in the algorithm.
This is new notion, although similar concepts have been used in previous scaling algorithms \cite{Goldfarb96,Vegh11}.
Section~\ref{sec:overview} explains the background and motivation of this notion.
Given a labeling $\mu$, let us call arcs in $E$ with $\gamma_{ij}^\mu<1$
{\em non-tight}, and denote their set by
\[
F^\mu:=\{ij\in E: \gamma_{ij}^\mu<1\}. 
\]
For every $i\in V$, let
\[
R_i:=\sum_{j: ji\in F^\mu}\gamma_{ji}f_{ji}
\]
denote the total flow incoming on non-tight arcs; let $R_i^\mu:=\frac{R_i}{\mu_i}=\sum_{j:ji\in F^\mu}\gamma_{ji}^\mu f_{ji}^\mu$.
For some $\Delta\ge 0$, let us
define the {\em $\Delta$-fat graph} as 
\[
E^\mu_f(\Delta)=E\cup \{ij: ji\in E, f^\mu_{ji}> \Delta\}.
\]
We say that $\mu$ is a {\em $\Delta$-conservative labeling} for $f$,
or that $(f,\mu)$ is a {\em $\Delta$-feasible pair}, if
\begin{itemize}
\item $\gamma_{ij}^\mu\le 1$ holds for all $ij\in E^\mu_f(\Delta)$, and 
\item $\mu_\sink=1$, and $\mu_i>0$, $e_i\ge R_i$  for every $i\in V-\sink{}$.
\end{itemize}

Note that in particular, $\mu$ must be  a feasible solution to (\ref{dual}).
The first condition is equivalent to requiring $f_{ij}^\mu\le \Delta$
for every non-tight arc.
Note that $0$-conservativeness is identical
to conservativeness: $E^\mu_f(0)=E^\mu_f$, and therefore 
every arc carrying positive flow must be tight;
the second condition simply gives $e_i\ge 0$ whenever $\mu_i>0$. 
The next lemma can be seen as the converse of this observation.

\begin{lemma}\label{lem:make-conservative}
Let $(f,\mu)$ be a $\Delta$-feasible pair for some $\Delta>0$. Let us
define the generalized flow $\tilde f$ with $\tilde f_{ij}=0$ if
$ij\in F^\mu$ and $\tilde f_{ij}=f_{ij}$ otherwise. Then $\tilde f$ is
a feasible generalized flow,  $\mu$ is a conservative labeling for $\tilde f$, and $Ex^\mu(\tilde f)\le Ex^\mu(f)+|F^\mu|\Delta$.
\end{lemma}
\begin{proof}
 For feasibility, we 
need to verify $e_i(\tilde f)\ge 0$ for all $i\in V-\sink$. This follows since
\[
e_i(\tilde f)\ge e_i(f)-\sum_{j: ji\in F^\mu}\gamma_{ji}f_{ji}=e_i(f)-R_i\ge 0.
\]
It is straightforward by the construction that  $\gamma_{ij}^\mu\le 1$
for every $ij\in E$ with equality whenever $\tilde f_{ij}>0$. This
shows that $\mu$ is a conservative labeling.
For the last part, observe that decreasing the flow value to 0 on a non-tight arc $ij$ may create $f^\mu_{ij}\le \Delta$ units of relabeled excess at $i$.
\end{proof}

\begin{claim}\label{cl:R-i-bound}
In a $\Delta$-conservative labeling, $R^\mu_i<\ind{i}\Delta$ holds for every $i\in V$.
\end{claim}
\begin{proof}
If $\mu$ is a $\Delta$-conservative labeling, then $f^\mu_{ji}\le
\Delta$ holds for every non-tight arc $ji$; also note that the relabeled flow arriving  from $j$ on a non-tight arc is 
$\gamma_{ji}^\mu f_{ji}^\mu<f_{ji}^\mu\le\Delta$, and hence $R^\mu_i<\ind{i}\Delta$. 
\end{proof}

\subsection{Overview of the algorithms}\label{sec:overview}
We now informally describe some fundamental ideas of our algorithms
\textsc{Continuous Scaling} and \textsc{Enhanced Continuous Scaling},
and explain their relations to previous generalized flow
algorithms. The precise algorithms and arguments will be given in the
later sections.

\subsubsection*{Basic features of the algorithms}
Given a generalized flow $f$, a  cycle $C$ in the residual graph $E_f$
is called {\em flow-generating}, if $\gamma(C)=\prod_{e\in
  C}\gamma_e>1$. If there exists a flow-generating cycle, then some
positive amount of flow can be sent around it to create positive
excess in an arbitrary node $i$ incident to $C$.

The notion of conservative labellings is closely related to flow
generating cycles. Notice that for an arbitrary labeling $\mu$,
$\gamma(C)=\gamma^\mu(C)$. Therefore, if $\mu$ is a finite
conservative labeling, then $E_f$ cannot contain any flow-generating cycles. It is also easy to verify the converse: if there are no flow-generating cycles, then there exists a conservative labeling (see also Lemma~\ref{lem:init}). 

The {\sc Maximum-mean-gain cycle-canceling} procedure, introduced by
Goldberg et al. \cite{Goldberg91},
can be used to eliminate all flow-generating cycles efficiently.
The subroutine proceeds by choosing a cycle $C\subseteq E_f$ maximizing $\gamma(C)^{1/|C|}$, and from an arbitrary node $i$ incident to $C$, 
sending the maximum possible amount of flow around $C$ admitted by the capacity constraints, thereby increasing the excess $e_i$. 
It terminates once there are no more flow-generating cycles left in $E_f$.
 This is a natural analogue of
the minimum mean cycle cancellation algorithm of Goldberg and Tarjan
\cite{Goldberg89} for minimum cost circulations. Radzik
\cite{Radzik93}  (see also \cite{Shigeno07}) gave a strongly polynomial running time bound
$O(m^2 n \log^2 n)$ for the {\sc Maximum-mean-gain cycle-canceling} algorithm.

Our algorithm also starts with performing this algorithm, with the
input being the initial solution $\bar f$ provided by
(\ref{cond:init}). Hence one can obtain a feasible solution $f$ along with a conservative labeling $\mu$ in strongly polynomial time.

Such an $f$ can be transformed to an optimal solution using Onaga's
algorithm \cite{Onaga67}: while there exists a node $i\in V-\sink$
with $e_i>0$, find a {\em highest gain augmenting path} from $i$ to
$\sink$, that is, a path $P$ in the residual graph $E_f$ with the product
of the gains maximum.  Send the maximum amount of flow on this
augmenting path enabled by the capacity constraints.
A conservative labeling can be used to identify
such paths: we can transform a conservative labeling to a {\em
  canonical labeling} (see \cite{Goldberg91}), where every node $i$ is
connected to the sink via a tight  path. Such a canonical labeling can
be  found via a Dijkstra-type algorithm, increasing the labels of certain nodes.
The correctness of Onaga's algorithm follows by the observation  that sending flow on a tight path maintains the conservativeness of the labeling, hence no new flow-generating cycles may appear.

Unfortunately, Onaga's algorithm may run in exponentially many steps,
and  might not even terminate
if the input is irrational.
The \textsc{Fat-Path} algorithm \cite{Goldberg91} introduces a scaling technique to overcome this difficulty.
The algorithm maintains a scaling factor $\Delta$ that decreases
geometrically. In the $\Delta$-phase, flow is sent on a highest gain
``$\Delta$-fat'' augmenting path, that is, a highest gain path among
those that have sufficient capacity to send $\Delta$ units of flow to
the sink. In our notation, these are paths in $E_f^\mu(\Delta)$. 
However, path augmentations might create new flow-generating cycles,
which have to be repeatedly cancelled by calling the cycle-canceling subroutine at the beginning of every phase.

Our notion of $\Delta$-feasible pairs in Section~\ref{sec:cons-label}
is motivated by the idea of $\Delta$-fat paths: note that every arc in
the $\Delta$-fat graph  $E_f^\mu(\Delta)$  has sufficient capacity to send $\Delta$ units of relabeled flow.
A main step in our algorithm will be sending $\Delta$ units of
relabeled flow on a tight path in $E_f(\Delta)$ from a node with
``high'' excess to the sink $\sink$ or  another node with ``low'' excess.
This is in contrast to \textsc{Fat-Path} and most other algorithms,
where these augmenting paths always terminate in the sink $\sink$.
We allow other nodes as well in order to maintain $e_i\ge R_i$ are
througout the algorithm. The purpose of these conditions is to
make sure that we always stay ``close'' to a conservative labeling: recall Lemma~\ref{lem:make-conservative} asserting that
 if $(f,\mu)$ is a
$\Delta$-feasible pair, then if we set the flow values to 0 on every
non-tight arc, the resulting $\tilde f$ is a feasible solution to
(\ref{primal}) not containing any flow-generating cycles. That is the
reason why we
  need to call the cycle-canceling algorithm only once, at the
initialization, in contrast to \textsc{Fat-Path}.

Similar ideas have been already used previously. The algorithm of Goldfarb, Jin and Orlin \cite{Goldfarb96} also uses a single initial cycle-canceling and then performs highest-gain augmentations in a scaling framework, combined with a clever bookkeeping on the arcs.
The algorithm in \cite{Vegh11} does not perform any cycle
cancellations and uses a homonymous notion of
$\Delta$-conservativeness that is closely related to ours; however, it
uses a different problem setup (called ``symmetric formulation''), and includes a condition stronger than $e_i\ge R_i$.

\subsubsection*{The way to the strongly polynomial bound}
The basic principle of our strongly polynomial algorithm is motivated
by Orlin's strongly polynomial algorithm for minimum cost circulations
(\cite{Orlin93}, see also \cite[Chapters 10.6-7]{amo}). The true
purpose of the algorithm is to compute a dual optimal solution to
(\ref{dual}). Provided a dual optimal solution, we can compute a
primal optimal solution to (\ref{primal}) by a single maximum flow
computation on the network of tight arcs (see
Theorem~\ref{thm:tight-flow}(i)).

The main measure of progress is identifying an arc $ij\in E$ that
must be tight in every dual optimal solution. Such an arc can be
contracted, and an optimal dual solution to the contracted instance
can be easily extended to an optimal dual solution on the original
instance (see Sections~\ref{sec:abundant}, \ref{sec:dual-prop}).
The algorithm could be simply restarted from scratch in the contracted
instance. Our algorithm \textsc{Enhanced Continuous Scaling} is
somewhat more complicated and keeps the previous primal solution 
to achieve better running time bounds by a
global analysis of all contraction phases.

We use a scaling-type algorithm to identify such arcs tight in every dual
optimal solution.
Our algorithm maintains a scaling parameter $\Delta$, and a
$\Delta$-feasible pair $(f,\mu)$ such that $Ex^\mu(f)\le 16m\Delta$.
Using standard flow decomposition techniques,
it can be shown that an arc $ij$ with $f^\mu_{ij}\ge\arcbound$ must be
positive in some optimal solution $f^*$ to (\ref{primal}) (see Theorem~\ref{thm:close-opt}). Then by
primal-dual slackness it follows that this arc is tight in every dual
optimal solution. Arcs with $f^\mu_{ij}\ge\arcbound$ will be called {\em
  abundant}.

A simple calculation (see the proof of Claim~\ref{cl:whenabundant-mod}) shows that once 
$|b_i^\mu|\ge 32mn\Delta$ for a node $i\in V-\sink$, there must be an
abundant arc leaving or entering $i$. Hence our goal is to design
an algorithm where such a node appears within a strongly polynomial
number of iterations.

A basic step in the scaling approaches
(e.g. \cite{Goldberg91,Goldfarb96,Vegh11}) is sending $\Delta$ units
of relabeled flow on a tight path; we shall call this a path augmentation.
In all previous approaches, the scaling factor $\Delta$ remained fixed
for a number of  path augmentations, and reduces by a substantial
amount (by at least a factor of two) for the next $\Delta$-phase.
Our main idea is what we call {\em continuous scaling}: the
boundaries between $\Delta$-phases are dissolved, and the scaling
factor decreases continuously, even during the iterations that lead to
finding the next path for augmentation. The precise description will
be given in Section~\ref{sec:weak}; in what follows, we give a
high-level overview of some key features.

We shall have a set $T_0$ with nodes of ``high'' relabelled excess;
 another set $N$ will
be the set consisting of the sink $\sink$ and further nodes with ``low'' relabelled excess.
 We will look for tight paths connecting a node in
$T_0$ to one in $N$; we will send $\Delta$ units of relabeled flow along
such a path.
In an intermediate elementary step, we let $T$ to denote the set of
nodes reachable from $T_0$ on a tight path; if it does not intersect
$N$, then we increase the labels $\mu_i$ for all $i\in T$ by the
same factor $\alpha$ hoping
that a
new tight arc appears between $T$ and $V\setminus T$, and thus $T$ can
be extended. We simultaneously
decrease the value of $\Delta$ by the same factor $\alpha$. Thus the relabeled
excess of nodes in $V\setminus T$ increases relative to $\Delta$. 
This might lead to
changes in the sets $T_0$ and $N$; hence an elementary step does not
necessarily terminate when a new tight arc appears, and therefore the value of $\alpha$ has to be carefully chosen.

This framework is undoubtedly more complicated than the traditional
scaling algorithms. The main reason for this approach is the phenomenon one might call ``inflation'' in the previous scaling-type algorithms.
There it might happen that the relabeling steps used for identifying
the next augmenting paths increase some labels by very high amounts,
and thus the relabeled flow remains small compared to $\Delta$ on
every arc of the network - therefore a new abundant arc can never be identified.
 It could even be the case that most $\Delta$-scaling phases do not perform any path augmentations at all, but only label updates: the relabeled excess at every node becomes smaller than $\Delta$ during the relabeling steps.\footnote{However,  to the extent of the author's knowledge, no actual examples are known for these phenomena in any of the algorithms.} 

The advantage of changing $\Delta$ continuously in our algorithm is
that the ratios $|b_i^\mu|/\Delta$ are nondecreasing for every
$i\in V-\sink$ during the entire algorithm. In the above described
situation, these ratios are unchanged for $i\in T$ and increase for
$i\in V\setminus T$. As remarked above, there must be an abundant arc
incident
to $i$ once $|b_i^\mu|/\Delta\ge 32mn$.

We first present a simpler version of this algorithm,
\textsc{Continuous Scaling} in Section~\ref{sec:weak}, where we can
only prove a weakly polynomial running time bound.
Whereas the ratios
$|b_i^\mu|/\Delta$ are nondecreasing, we are not able to
prove that one of them eventually reaches the level $32mn$ in a
strongly polynomial number of steps.
This is since the set $V\setminus T$
where the ratio increases might always consist only of nodes where
$|b_i^\mu|/\Delta$ is very small. The algorithm \textsc{Enhanced
  Continuous Scaling} in Section~\ref{sec:strong} therefore introduces one additional subroutine,
called \textsc{Filtration}. In case $|b_i^\mu|<\Delta/n$ for every
$i\in (V\setminus T)-\sink$,  we ``tidy-up'' the flow inside
$V\setminus T$, by performing a maximum flow computation here. 
This drastically reduces all relabeled excesses in $V\setminus T$, and
thereby guarantees that most iterations of the algorithm will have to
increase certain $|b_i^\mu|/\Delta$ values that are already at least $1/n$.

In summary, the strongly polynomiality of our algorithm is based on the following three main new ideas.
\begin{itemize}
\item The definition of $\Delta$-feasible pairs, in particular, the condition on maintaining a security reserve $R_i$. It is a cleaner and more efficient framework than similar ones in \cite{Goldfarb96} and \cite{Vegh11}; we believe this is the ``real'' condition a scaling type algorithm has to maintain.
\item Continuous scaling, that guarantees that the ratios $|b_i^\mu|/\Delta$ are nondecreasing during the algorithm. This is achieved by doing the exact opposite of \cite{Goldberg91,Goldfarb96,Vegh11} that use the natural analogue of the scaling technique for minimum cost circulations.
\item The \textsc{Filtration} subroutine that intervenes in the algorithm whenever the nodes on a certain, relatively isolated part of the network
have ``unreasonably high'' excesses as compared to the small node demands in this part.
\end{itemize}

\subsection{The maximum flow subroutine}\label{sec:tight-flow}
Standard maximum flow computation (see e.g. \cite[Chapters 6-7]{amo})
will be a crucial subroutine in our algorithm. First and foremost, if
an optimal labeling is provided, then an optimal solution to (\ref{primal}) can be obtained by computing a maximum flow. 
We now describe the subroutine  \textsc{Tight-Flow}$(S,\mu)$, to
perform such computations. In the weakly polynomial algorithm
(Section~\ref{sec:weak}), it will be used only twice: at the
initialization and at  the termination of the algorithm, in both cases
with $S=V$. However, it will also be the key part of the subroutine  \textsc{Filtration}  in the strongly polynomial algorithm (Section~\ref{sec:strong}), also applied to subsets $S\subsetneq V$.


The input of \textsc{Tight-Flow}$(S,\mu)$ is a node
set $S\subseteq V$ with $\sink\in S$,  and a  labeling $\mu$, that is a feasible solution to (\ref{dual}) when restricted to $S$.
The subroutine returns a generalized flow $f'$ supported on $E[S]$, such that
$\mu$ restricted to $S$ is a  conservative  labeling for $f'$.
Let us define the arc set $\tilde E\subseteq E[S]$ as the set
of tight arcs for $\mu$:
\[
\tilde E:=\{ij\in E[S]:  \gamma^\mu_{ij}=1\}.
\]
Let us extend $S$ by a new source node $s$, and
add an arc $si$ from $s$ to  every $i\in S-\sink$; let $\tilde E'$ denote the union of $\tilde E$ and these new arcs.
Let us set lower and upper arc capacities $\ell_{ij}:=0$, $u_{ij}:=\infty$
on all arcs of $\tilde E$; for $i\in S-\sink$,
let $\ell_{si}:=-\infty$ and $u_{si}:=-b_i^\mu$.

\textsc{Tight-Flow}$(S,\mu)$ computes a maximum flow $x$ from $s$ to $\sink$ on the
network $(S\cup\{s\},\tilde E')$ with capacities $\ell$ and $u$. Let us define
$f':E[S]\rightarrow \R_{\ge0}$ by $f'_{ij}:=x_{ij}\mu_i$ if $ij\in \tilde E$
and $f'_{ij}:=0$ otherwise. This completes the description of the subroutine
\textsc{Tight-Flow}.
Because of the possibly negative upper capacities on the $si$ arcs, the maximum flow problem might be infeasible; in this case, the subroutine returns an error.

\begin{theorem}\label{thm:tight-flow}
\begin{enumerate}[(i)]
\item If $\mu$ is an optimal solution to (\ref{dual}), then \textsc{Tight-flow}$(V,\mu)$ returns an optimal solution to (\ref{primal}).
\item 
Assume that the maximum flow problem in \textsc{Tight-flow}$(S,\mu)$ is
feasible, and returns a vector $f'$. Then $f'$ is a feasible solution to (\ref{primal}) on $S$, and
\[
e^\mu_i(f')\le n\max_{j\in S-t}|b_j^\mu|\quad \forall i\in S.
\]
\item  Assume that the flow problem in \textsc{Tight-flow}$(V,\mu)$ is feasible  and returns a generalized flow $f'$ with $Ex(f')<1/{\bar B^3}$. Then $Ex(f')=0$ must hold, that is, $f'$ is an optimal solution to (\ref{primal}).
\end{enumerate}
\end{theorem}
\begin{proof}
To prove part {\em (i)}, assume $\mu$ is an optimal labeling. Let $g$ be an optimal solution to (\ref{primal}). Let us define
$x_{ij}:=g^\mu_{ij}$ if $ij\in E$ and $x_{si}:=\sum_{j:ij\in  E} g^\mu_{ij}-\sum_{j:ji\in E} g^\mu_{ji}$ for every $i\in V-\sink$.
By Theorem~\ref{thm:genflow-opt}(i), $x_{si}=-b_i^\mu$ for all $i\in
V-\sink$, and therefore $x$
is a maximum flow, with $(\{s\},V)$ forming a minimum cut. Conversely, an arbitrary maximum flow must saturate
every arc leaving $s$, and therefore we get $e_i(f')=0$ for every
$i\in V-\sink$ for the $f'$ returned by \textsc{Tight
  Flow}$(V,\mu)$. It is straightforward that all conditions in
Theorem~\ref{thm:genflow-opt}(i) are satisfied.

For part {\em (ii)}, first observe that if there is a feasible
solution $x$ to the flow problem, then $e_i(f')\ge 0$ must hold for
every $i\in V-\sink$, due to the constraint $x_{si}\le -b_i^\mu$; further,
$\mu$ is a conservative labeling for $f'$.
Let us pick a node  $r\in S-\sink$ with  $e_r(f')>0$, and
let $Z\subseteq S$ denote the set of nodes that can be reached from $r$ on a 
directed path in the residual graph $(S,\tilde E_{f'})$, defined by 
\[\tilde E_{f'}=\tilde E\cup\{ji: ij\in \tilde E, f'_{ij}>0\}.\]
 Note that ${f'_{ij}}^\mu=x_{ij}$ for every $ij\in \tilde E_{f'}$.
Assume that  $\sink\in Z$, that is, there is a directed path $P$ from $r$ to $\sink$ in the residual graph. Since $e_r(f')>0$, we have 
$x_{sr}<-b^\mu_i=u_{sr}$; hence $sr$ and $P$ give an augmenting path for the flow $x$, in 
a contradiction to its choice as a maximum flow.

We may thus conclude that  $\sink\notin Z$. Hence $e^\mu_i(f')\ge 0$ for all $i\in Z$, and therefore
\begin{equation}\label{eq:S-sum}
0< e^\mu_r(f')\le \sum_{i\in Z}e^\mu_i(f')=\sum_{i\in Z}\left(\sum_{j\in Z:ji\in
    \tilde E}x_{ji}-\sum_{j\in Z: ij\in
    \tilde E}x_{ij}-b_i^\mu\right)=-\sum_{i\in
  Z}b_i^\mu\le n\max_{j\in S-t}|b_j^\mu|,
\end{equation}
proving part {\em (ii)} of the Theorem. Here we used that 
if $x_{ij}>0$ then $i\in Z$ if and only if $j\in Z$.

Let us turn to part {\em (iii)}; assume  that $e^\mu_r(f')>0$ for some $r\in V-\sink$.  The equation (\ref{eq:S-sum}) can be further written as
\begin{equation}\label{eq:S-sum-2}
0<e^\mu_r(f')\le\sum_{i\in Z}e^\mu_i(f')= -\sum_{i\in
  Z}b_i^\mu= -\frac{1}{\mu_{r}}\sum_{i\in Z}b_i\frac{\mu_{r}}{\mu_{i}}.
\end{equation}
For every $i\in Z$, there is a tight path $P$ in $\tilde E_{f'}$ from $r$ to $i$, that is,
${\mu_{r}}/{\mu_{i}}=\prod_{e\in P}1/\gamma_e$.
 By our assumption on the encoding sizes, this product must be an integer multiple of $1/\bar B$.
We further assumed that every $b_i$ value is an integer multiple of $1/\bar B$.
Hence  every
term $b_i\frac{\mu_{r}}{\mu_{i}}$ is an integer multiple of ${1}/{\bar B^2}$.
Further, by (\ref{cond:root}), we have $r\sink\in E$, and
$\gamma_{r\sink}\ge 1/\bar B$.
 By the conservativeness of $\mu$
w.r.t. to $\tilde f$,
 $\frac{1}{\mu_{r}}\ge \gamma_{r\sink}\ge 1/\bar B$. 
 Consequently, the last expression in  (\ref{eq:S-sum-2}) must be at least 
$1/\bar B^{3}$ whenever it is nonzero. Therefore 
\[
1/\bar B^3\le \sum_{i\in Z}e^\mu_i(f')\le Ex^\mu(f'),
\]
contradicting our assumption. Hence it follows that  $e_r(f)=0$ for all $r\in V-\sink$.
\end{proof}

section{The Continuous Scaling algorithm}\label{sec:weak}

\begin{figure}[htb]
\begin{center}
\fbox{\parbox{\textwidth}{
\begin{tabbing}
xxxxx \= xxx \= xxx \= xxx\= xxx \= xxx \= xxxxxxxxxxx \= \kill
\> \textbf{Algorithm} \textsc{Continuous Scaling}\\
\> \textsc{Initialize}$(V,E,b,\gamma,\bar f)$ ; \\
\> $T_0\leftarrow \emptyset$ ; $T\leftarrow \emptyset$ ;\\
\> \textsc{While} $\Delta\ge\threshold$ \textsc{do}\\
\> \> $N\leftarrow\{\sink\}\cup \{i\in V-\sink:
e_i^\mu<\negd{i}\Delta\}$ ;\\
\>  \> \textbf{if} $N\cap T\neq \emptyset$ \textbf{then}\\
\> \> \>  \textbf{pick} $p\in T_0$, $q\in N\cap T$ connected by a  tight path $P$ in $E^\mu_f(\Delta)$ ;\\
\> \> \> \textbf{send} $\Delta$ units of relabeled flow from $p$ to
$q$ along $P$ ;\\
\> \> \> \textbf{if} $e_p^\mu<\egyd{p}\Delta$ \textbf{then} $T_0\leftarrow
T_0\setminus \{p\}$ ;\\
\> \> \> $T\leftarrow T_0$ ;\\
\> \> \textbf{else}\\
\> \> \> \textbf{if} $\exists ij\in E_f^\mu(\Delta)$,
$\gamma_{ij}^\mu=1$, $i\in T$, $j\in V\setminus T$ \textbf{then}
$T\leftarrow T\cup \{j\}$ ;\\
\> \> \> \textbf{else} \textsc{Elementary Step}$(T,T_0,f,\mu,\Delta)$ ;\\
\> \textsc{Tight-Flow}$(V,\mu)$ ; 
\end{tabbing}
}}
\caption{Description of the weakly polynomial algorithm}\label{code:weak}
\end{center}
\end{figure}

The algorithm \textsc{Continuous Scaling} is shown in
Figure~\ref{code:weak}. 
The strongly polynomial algorithm \textsc{Enhanced
  Continuous Scaling} in Section~\ref{sec:strong} will be an improved variant of this.
We shall always make assumptions
(\ref{cond:root}), (\ref{cond:init}) and (\ref{cond:bounded}).

The algorithm starts with the subroutine \textsc{Initialize},
described in Section~\ref{sec:init}, that returns an initial flow $f$, along with a
$\Delta=\bar\Delta$-conservative
labeling $\mu$ such that $e_i^\mu<\egyd{i}\Delta$ holds for every
$i\in V$.  This is based on the 
{\sc Maximum-mean-gain cycle-canceling} algorithm as in \cite{Goldberg91,Radzik93}.
The main part of the algorithm  (the while loop) consists of iterations.
The value of the scaling parameter $\Delta$ is monotone decreasing, 
and all $\mu_i$ values are monotone increasing during the algorithm.
In every iteration, a $\Delta$-feasible pair $(f,\mu)$ is maintained.
These iterations  stop once the scaling parameter $\Delta$ decreases below \threshold.
At this point we apply the subroutine \textsc{Tight-flow}$(V,\mu)$, as described in Section~\ref{sec:tight-flow}, to find an optimal solution by a single maximum flow computation.

 The set $N$  denotes the set consisting of $\sink$ and all nodes with
$e_i^\mu<\negd{i}\Delta$. The set $T_0$ consists of a certain set of nodes (but not all) with
$e_i^\mu\ge\egyd{i}\Delta$. The set $T$  denotes a set of nodes
that can be reached from $T_0$ on a  tight path in the
$\Delta$-fat graph $E^\mu_f(\Delta)$. Both $T_0$ and $T$ are
initialized empty.

Every iteration first checks
whether $N\cap T\neq\emptyset$. If yes, then nodes $p\in
T_0$ and $q\in N\cap T$ are picked connected by a tight path $P$ in the $\Delta$-fat
graph.
  $\Delta$ units of relabeled flow is sent from $p$ to $q$ on
$P$: that is, $f_{ij}$ is increased by $\Delta\mu_i$ for every $ij\in
P$ (if $ij$ was a reverse arc, this means decreasing $f_{ji}$ by
$\Delta\mu_j$). The only $e_i$ values that change are $e_p$ and $e_q$.
If the new value is $e_p^\mu<\egyd{p}\Delta$, then $p$ is removed from $T_0$.
The iteration finishes in this case by resetting $T=T_0$
(irrespective to whether $p$ was removed or not).

Let us now turn to the case $N\cap T=\emptyset$. If there is a node
$j\in V\setminus T$ connected by a tight arc in $E^\mu_f(\Delta)$ to
$T$, then we extend $T$ by $j$, and the iteration
terminates. Otherwise, the subroutine \textsc{Elementary Step}$(T,T_0,f,\mu,\Delta)$ is
called. The precise description is given in Section~\ref{sec:elementary};
 we give an outline below.

 For a carefully chosen
$\alpha>1$, all $\mu_i$ values are multiplied by $\alpha$ for $i\in
T$, and  $\mu_i$ is left unchanged for $i\in V\setminus T$. At the same
time, $\Delta$ is divided by $\alpha$ (this is the only step in the
main part of the
algorithm modifying the $\mu_i$'s and the value of $\Delta$). The
flow is divided by $\alpha$ on all non-tight arcs
in $F^\mu[V\setminus T]$, and on every arc entering $T$.
The value of  $\alpha$ is chosen to be the largest such that the
labeling remains $\Delta$-feasible with the above changes, and further
$e_i^\mu\le \ketd{i}\Delta$ holds for all $i\in V\setminus T$.
If  $\alpha=\infty$, then the algorithm terminates with an optimal solution.
For finite $\alpha$,  all 
nodes $i$ for which $e_i^\mu= \ketd{i}\Delta$  holds after the
change are added both to $T_0$ and to $T$.
On the other hand, the $e_i^\mu$ values might also decrease both for $i\in T$ and $i\in V\setminus T$. If for some $i\in T_0$, the value of $e_i^\mu$ drops below 
$\egyd{i}\Delta$, then $i$ is removed from $T_0$, and
$T$ is reset to $T=T_0$. 
In every step when $T_0$ is not changed, a tight arc in $E^\mu_f(\Delta)$
leaving $T$ must appear. Consequently, $T$ will be extended in the next iteration.
We shall prove the following running time bound:
\begin{theorem}\label{thm:weak-running}
The algorithm \textsc{Continuous Scaling} can be implemented to find an optimal solution
for the uncapacitated formulation (\ref{primal}) in
running time $\max\{O(m(m+n\log n) \log \bar B), O(m^2n\log^2 n)\}$.
\end{theorem}
The high level idea of the analysis is the following. 
The $e_i^\mu$ values for nodes $i\in T_0$ are non-increasing, and a path augmentation
starting from $i$ reduces $e_i^\mu$ by $\Delta$. The node $i$ leaves $T_0$ once $e_i^\mu$ drops below $\egyd{i}\Delta$, and may enter
again once it increases to $\ketd{i}\Delta$. As shown in Lemma~\ref{lem:sigma-bound}, the value of $\Delta$ must decrease by at least a factor 2 between two such events. Also, it is easy to verify that within every $2n$ \textsc{Elementary step} operations, either a path augmentation must be carried out, or a node $i$ must leave $T_0$ due to decrease in $e_i^\mu$ caused by label changes. These two facts together give a polynomial bound on the running time.

In the proof of Theorem~\ref{thm:weak-running}, we outline a more efficient implementation of the
algorithm, with all iterations between two path augmentations
performed together.

For a problem in the standard form on $n$ nodes, $m$ arcs and
complexity parameter $B$, the reduction in Section~\ref{sec:transform}
shows that it can be transformed to an equivalent instance with $\le n+m$
nodes, $\le 2m$ arcs, and $\bar B\le 2B^{4m}$. Hence the theorem gives a
running time $O(m^3\log n \log B)$, assuming $n\le B$.

Our algorithm could be naturally adapted to work on a problem
instance with both node demands and arc capacities; the reason for
choosing the uncapacitated instance is its suitability
for the strongly polynomial algorithm in Section~\ref{sec:strong}. 
Such a modification would run in time $O(m^2(m+n\log n)\log
B)$, matching the bound of Goldfarb et al. \cite{Goldfarb97}.


\subsection{The Initialization subroutine}\label{sec:init}
In this section we describe the \textsc{Initialize}$(V,E,b,\gamma,\bar f)$ subroutine.
The input is a graph $G=(V,E)$, node demands $b_i:V\rightarrow \R$, gain factors 
$\gamma:E\rightarrow \R_{>0}$ and the initial generalized flow $\bar f$
guaranteed by the  assumption (\ref{cond:init}).
The initial value of $\Delta=\bar\Delta$ is computed and 
a $\Delta$-feasible pair $(f,\mu)$ is returned such that $e_i^\mu<\egyd{i}\Delta$ holds for every
$i\in V-\sink$.

First, we use the {\sc Maximum-mean-gain cycle-canceling} algorithm by Radzik \cite{Radzik93}. This paper uses
the standard capacitated formulation \eqref{primal-standard} with
finite capacities on the arcs. As a consequence, every flow-generating
cycle can only generate a finite amount of flow. Our boundedness
assumption \eqref{cond:bounded}, together with \eqref{cond:root},
guarantees the same property. Provided this boundedness property, Radzik's strongly polynomial bound extends verbatim to the uncapacitated formulation \eqref{primal}.

This returns a generalized flow $g$ such that the residual graph
$E_g$ contains no flow generating cycles, that is, no cycles $C$ with $\gamma(C)>1$.
Let us define $\mu_\sink:=1$ and for $i\in V-\sink$,
\begin{equation}\label{def:mu}
\mu_i:=1/\max\left\{\gamma(P): P\subseteq E_g\mbox{ is a walk from $i$ to }\sink.\right\}
\end{equation}
Such a path must exist according to assumption (\ref{cond:root}), and
since $\gamma(C)\le 1$ for all cycles $C$, 
 the walk giving the maximum can always be chosen to be a path.
The $\mu_i$ values can be computed efficiently: note that they
correspond to shortest paths with respect to the cost function $-\log
\gamma_e$. To avoid computing logarithms, we may use a multiplicative version of Dijkstra's
algorithm to obtain the $\mu_i$ values in strongly polynomial time.

After the cycle cancelling subroutine and computing the $\mu_i$ values, the subroutine \textsc{Tight Flow}$(V,\mu)$ is called, as described in
Section~\ref{sec:tight-flow}. This computes a generalized flow
$f'$. We set $f=f'$, and set the initial $\Delta=\bar\Delta:=\max_{i\in V-\sink}e_i^\mu$.

\subsection{The  Elementary step subroutine}\label{sec:elementary}
\begin{figure}[htb]
\begin{center}
\fbox{\parbox{\textwidth}{
\begin{tabbing}
xxxxx \= xxx \= xxx \= xxx \= xxxxxxxxxxx \= \kill
\> \textbf{Subroutine} \textsc{Elementary step}$(T,T_0,f,\mu,\Delta)$\\
\> $\alpha_1\leftarrow\min\left\{\delta_i: i\in (V\setminus T)-\sink
\right\}$, with $\delta_i$ as defined in \eqref{eq:delta-i} ;\\
\> $\alpha_2\leftarrow\min\left\{\frac{1}{\gamma_{ij}^\mu}: ij\in
  E[T,V\setminus T]\right\}$ ;\\
\> $\alpha\leftarrow\min\{\alpha_1,\alpha_2\}$ ;\\
\> \textbf{if} $\alpha=\infty$ \textbf{then} \\
\> \> set $f_{ti}=0$ for all $ti\in E$: $\gamma^\mu_{ti}<1$ ;\\
\> \> \textbf{return} optimal flow $f$ and optimal relabeling $\mu$ ;\\
\> \> \textbf{TERMINATE}.\\
\> $\Delta\leftarrow \frac{\Delta}{\alpha}$ ;\\
\> \textbf{for} $i\in T$ \textbf{do} $\mu_i\leftarrow \alpha\mu_i$ ;\\
\> \textbf{for} $ij\in F^\mu[V\setminus T]\cup E[V\setminus
T,T]$ \textbf{do} $f_{ij}\leftarrow \frac{f_{ij}}{\alpha}$ ;\\
\> $T_0\leftarrow T_0\cup \{i: i\in V\setminus T,\ e_i^\mu= \ketd{i}\Delta\}$ ;\\
\> $T\leftarrow T\cup T_0$ ;\\
\> \textbf{if} $\exists i\in T_0:  e_i^\mu<\egyd{i}\Delta$ \textbf{then} \\
\> \> $T_0\leftarrow T_0\setminus \{i:e_i^\mu<\egyd{i}\Delta\}$ ;\\
\> \> $T\leftarrow T_0$ ;
\end{tabbing}
}}
\caption{The Elementary Step subroutine}\label{code:basic}
\end{center}
\end{figure}

Let $(f,\mu)$ be a $\Delta$-feasible pair for $\Delta>0$. Let
$T\subseteq V$ be a (possibly empty) set of nodes with $e_i^\mu\le
\ketd{i}\Delta$ for every 
$i\in V$, with strict inequality whenever
$i\in V\setminus T$. 
The subroutine  (Figure~\ref{code:basic}) perfoms the following
modifications for some $\alpha>1$.
The $\mu_i$ values are multiplied by $\alpha$
for $i\in T$, and left unchanged for $i\in V\setminus T$. The new
value of the scaling parameter is set to
$\Delta/\alpha$. Finally, the flow on non-tight arcs $ij\in
F^\mu[V\setminus T]$ and on all arcs $ij\in E[V\setminus T,T]$ is divided
by $\alpha$.

The value of $\alpha$ is chosen maximal such that for the new values
of $f,\mu$, and $\Delta$,  
$(f,\mu)$ is $\Delta$-feasible, and further the modified excess
$e_i\le \ketd{i}\Delta\mu_i$ holds for every $i\in V$. 
For the latter, we need the following definitions for every $i\in
(V\setminus T)-\sink$. Let 
\begin{align}
\begin{tabular}{ll}
$F_1(i):=\delta^{in}(i)\cap F^\mu[V\setminus T]$, &
$r_1(i):=\sum_{j:ji\in  F_1(i)}\gamma_{ji}f_{ji}$,\\
$F_2(i):=\delta^{in}(i)\setminus F_1(i)$, &
$r_2(i):=\sum_{j:ji\in   F_2(i)}\gamma_{ji}f_{ji}$, \\
$F_3(i):=\delta^{out}(i)\cap (F^\mu[V\setminus T]\cup E[V\setminus
T,T])$, &
$r_3(i):=\sum_{j:ij\in F_3(i)}f_{ij}$,\\
$F_4(i):=\delta^{out}(i)\setminus F_3(i)$,&
$r_4(i):=\sum_{j:ij\in F_4(i)}f_{ij}$.
\end{tabular}\label{def:f-r}
\end{align}
Note that $F_1(i)$ and $F_3(i)$ denote the set of those incoming and outgoing arcs
where we wish to  decrease the flow by a factor $\alpha$.
For every $i\in
(V\setminus T)-\sink$, let us define
\begin{equation}
\delta_i:=\frac{\ketd{i}\Delta\mu_i +r_3(i)-r_1(i)}
{r_2(i)-r_4(i)-b_i}.\label{eq:delta-i}
\end{equation}
If the denominator is 0 then $\delta_i:=\infty$ is set. We shall verify in the proof of Lemma~\ref{lem:elementary} that the denominator is always nonnegative and the numerator is positive.

The subroutine (Figure~\ref{code:basic}) chooses 
$\alpha$ as the largest value subject to $\alpha\le \delta_i$ for all $i\in
(V\setminus T)-t$, and $\alpha\le \frac{1}{\gamma_{ij}^\mu}$ for all
arcs $ij\in E$ leaving the set $T$.
If $\alpha=\infty$, then $f$ becomes an optimal solution, after setting the
value on all non-tight arcs leaving $\sink$ to $0$. 
If $\alpha$ is finite, the algorithm
  performs the above described modifications.
Nodes $i$ with $e_i^\mu=\ketd{i}\Delta$  (that is, $\alpha=\delta_i$) are added to both $T_0$ and $T$.
Finally, if  $e_i^\mu$ drops below 
$\egyd{i}\Delta$ for some $i\in T_0$, then all such nodes $i$ will be removed from $T_0$, and
$T$ is reset to $T=T_0$. 
The validity of this subroutine is proved in Lemma~\ref{lem:elementary}.

\section{Analysis of the Continuous Scaling algorithm}\label{sec:anal-weak}

\begin{lemma}\label{lem:init}
The subroutine \textsc{Initialize}$(V,E,b,\gamma,\bar f)$ returns a $\Delta$-feasible pair
$(f,\mu)$ with $e_i^\mu\le
\egyd{i}\Delta$ for every $i\in V-\sink$, and $\Delta=\bar\Delta\le \Deltabound$.
\end{lemma}
\begin{proof}
First, we have to verify that the flow problem in
\textsc{Tight-flow}$(V,\mu)$ is feasible. 
We use the generalized flow $g$ obtained by the {\sc Maximum-mean-gain cycle-canceling} algorithm to verify this, by showing that $\mu$ is a conservative labeling for $g$. The nontrivial part is to prove $\gamma_{ij}^\mu\le 1$ for every
residual arc
$ij\in E_g$. 

Let $ij\in E_g$ be an arbitrary residual arc. Consider the $j-\sink$
path $P^j$ with $\mu_j=1/\gamma(P^j)$ in (\ref{def:mu}). Let $P'$
denote the path resulting by adding the arc $ij\in E_g$ to the beginning of $P^j$. Then by definition, $1/\mu_i\ge \gamma(P')=\gamma_{ij}/\mu_j$, showing $\gamma_{ij}^\mu\le 1$. 

Let us now consider the maximum flow instance in \textsc{Tight-flow}$(V,\mu)$.
Setting $x_{ij}:=g^\mu_{ij}$ if $ij\in E$ and $x_{si}:=\sum_{j:ij\in
  E} g^\mu_{ij}-\sum_{j:ji\in E} g^\mu_{ji}$ for every $i\in V-\sink$
gives a feasible solution. This guarantees the existence of the
optimal solution $f'$.

It is straightforward by the construction that $\mu$ is a conservative labeling for $f$, and hence $(f,\mu)$ is $\Delta$-feasible for arbitrary $\Delta>0$. The condition 
 $e_i^\mu\le \egyd{i}\Delta$ is also straightforward by definition. 

Let us verify the bound on $\Delta$. By  Theorem~\ref{thm:tight-flow}(ii), we have 
$\Delta\le n\max_{i\in V-\sink}{|b_i|}/{\mu_i}$.
Our assumption on the encoding sizes give $|b_i|\le \bar B$. Further, 
we have  $1/\mu_i\le \bar B$, according to the definition of $1/\mu_i=\gamma(P^i)$ for some $i-\sink$ path $P^i$, and the encoding assumptions on the $\gamma_e$ values.
\end{proof}

The next straightforward claim justifies the path augmentation step
carried out between $p\in T_0$ and $q\in N\cap T$ whenever $N\cap T\neq\emptyset$. 
\begin{claim}\label{cl:path-aug}
Let $(f,\mu)$ be a $\Delta$-feasible pair, and assume $P$ is a tight path in
$E_f^\mu(\Delta)$ from node $p$ to node $q$, with $e_p^\mu(f)\ge
\Delta+R_p^\mu$. Let us increase $f_{ij}$ by $\Delta\mu_i$ if $ij\in P$ is a
forward arc, and decrease $f_{ji}$ by $\Delta\mu_j$ if $ij\in P$ is a
backward arc; let $f'$ denote the resulting flow. Then $(f',\mu)$ is
also a $\Delta$-feasible pair.
\end{claim}

We next prove some fundamental properties of the subroutine
\textsc{Elementary step}, most importantly, that it maintains 
the $\Delta$-feasibility of $(f,\mu)$. By induction, we may assume 
that the four conditions in the lemma always hold when \textsc{Elementary step}$(T,T_0,f,\mu,\Delta)$ is called in the algorithm.
\begin{lemma}\label{lem:elementary}
Let $(f,\mu)$ be a $\Delta$-feasible pair for some $\Delta>0$, and let $T\subseteq V-\sink$ satisfy the following conditions:
\begin{itemize}
\item $e_i^\mu< \ketd{i}\Delta$ for all $i\in (V\setminus T)-\sink$; 
\item $e_i^\mu\ge \negd{i}\Delta$ for all $i\in  T$; 
\item $\gamma^\mu_{ij}<1$ for all $ij\in E[T,V\setminus T]$;
\item $f_{ij}^\mu\le \Delta$ for all $ij\in E[V\setminus T,T]$.
\end{itemize}
Let $f'$, $\mu'$, $\Delta'$, and $e'_i$ denote the respective values at the
end of \textsc{Elementary step}$(T,T_0,f,\mu,\Delta)$.
If $\alpha<\infty$, then the pair $(f',\mu')$ is
$\Delta'$-feasible. Further, the following statements hold.
\begin{enumerate}[(i)]
\item $\alpha>1$. If $\alpha=\infty$ then the modified flow returned
  by the algorithm is optimal to  (\ref{primal}), and $\mu$ is optimal  to  (\ref{dual}).
\item ${e'_i}^{\mu'_i}\le  \ketd{i}\Delta'$ for all $i\in V\setminus T$, and if
 $\alpha=\alpha_1$, then $\exists i\in V\setminus T$ such that
 equality holds.
\item $e'_i\le e_i$ for all $i\in T$.
\item If $\alpha=\alpha_2$ then $\exists ij\in E$ with $i\in T$, $j\in V\setminus T$, and $\gamma_{ij}^{\mu'}=1$.
\end{enumerate}
\end{lemma}
\begin{proof}
For $\Delta'$-feasibility, let us first verify $\gamma_{ij}^{\mu'}\le 1$
for all $ij\in E$. If $ij\in E[T]$ or $ij\in
 E[V\setminus T]$, then $\gamma_{ij}^{\mu'}=\gamma_{ij}^\mu$.
If $ij\in E[T,V\setminus T]$, then we have
$\gamma_{ij}^{\mu'}=\alpha\gamma_{ij}^\mu\le 1$ due to the choice
$\alpha\le \alpha_2$.
Finally, if $ij\in E[V\setminus T,T]$, then
$\gamma_{ij}^{\mu'}=\gamma_{ij}^\mu/\alpha<1$. The next two claims
verify the remaining properties needed for $\Delta'$-feasibility.
\begin{claim}
If  $\gamma_{ij}^{\mu'}<1$ for an arc $ij\in E$, then ${f'_{ij}}^{\mu'}=f_{ij}^\mu/\alpha\le \Delta/\alpha=\Delta'$.
\end{claim}
\begin{proof}
Let us first assume  $i\in T$;   the first equality follows by $f'_{ij}=f_{ij}$,
$\mu'_{i}=\mu_i\alpha$.
The inequality $f_{ij}^\mu\le \Delta$ is due to the
$\Delta$-feasibility of $f$, because of $\gamma_{ij}^\mu<1$.
 If $j\in V\setminus T$, this is included among the assumptions,
 whereas if $j\in T$, then it follows by $\gamma_{ij}^{\mu}=\gamma_{ij}^{\mu'}<1$.

Consider now the case $i\in V\setminus T$. If also $j\in V\setminus T$,
then $\gamma_{ij}^\mu=\gamma_{ij}^{\mu'}<1$, and hence
$f'_{ij}=f_{ij}/\alpha$, as we decrease the flow values by a factor
$\alpha$ on arcs $F^\mu[V\setminus T]$; the inequality $f_{ij}^\mu\le
\Delta$  follows again
by the $\Delta$-feasibility of $f$. If $j\in T$, that is, $ij\in
E[V\setminus T,T]$, then we must again have $f'_{ij}=f_{ij}/\alpha$, 
and $f_{ij}^\mu\le \Delta$ is included among the assumptions.
\end{proof}
\begin{claim}
The inequality $e'_i\ge R'_i$ holds for all $i\in V-\sink$, where  $R'_i$ denotes the
 $f'$ flow entering $i$ on non-tight arcs for $\mu'$. 
\end{claim}
\begin{proof}
 We have $e_i\ge
R_i$  by the $\Delta$-feasibility of $f$.

\smallskip
\noindent{\bf Case I: $i\in V\setminus T$.} Since $f'\le f$, the change of
flow on outgoing arcs may only increase $e_i$. 
 If $f'_{ji}<f_{ji}$ on
an incoming arc $ji\in E$, then 
$j\in V\setminus T$ must hold.
Therefore $\gamma_{ji}^{\mu'}=\gamma_{ji}^\mu$, and hence 
$ji$ must be a non-tight arc for both
$\mu$ and  $\mu'$.
The change on $ji$ decreases $e_i$ by
$(1-1/\alpha)\gamma_{ji}f_{ji}$, and causes the same change in
the value of $R_i$. 

\noindent{\bf Case II: $i\in T$.} By the assumption of the lemma,
$e_i^\mu\ge \negd{i}\Delta$.
The flow on outgoing arcs is unchanged.
Let $ji\in E$ be an incoming arc with $f'_{ji}<f_{ji}$. 
We must have $j\in V\setminus T$ and thus  $f_{ji}\le \Delta\mu_j$ by
assumption; further, $\gamma^\mu_{ji}\le 1$ by the
$\Delta$-feasibility of $(f,\mu)$.
Hence it follows that  $\gamma_{ji}f_{ji}<\Delta\gamma_{ji}\mu_j\le \Delta\mu_i$.
This enables us to bound the value $e'_i$. Let $\lambda$ denote the number
of arcs $ji$ with $j\in V\setminus T$.
Using also the assumption $e_i\ge\negd{i}\Delta\mu_i$, we have
\begin{align*}
e'_i&= e_i- \sum_{j\in V\setminus T:ji\in E}(\gamma_{ji}f_{ji}- \gamma_{ji}f'_{ji})
\ge e_i-\sum_{j\in V\setminus T:ji\in E} (\Delta\mu_i-\gamma_{ji}f'_{ji})\\
& 
\ge \sum_{j\in V\setminus T:ji\in E} \gamma_{ji}f'_{ji}+(\ind{i}+1-\lambda)\Delta\mu_i>  R'_i.
\end{align*}
In the last inequality, we use that if $ji$ is a non-tight arc with $j\in
T$, then $\gamma_{ji}f'_{ji}\le \Delta'\mu'_i=\Delta\mu_i$, 
 and that the total number of such
arcs is $\le \ind{i}-\lambda$. 
\end{proof}
Let us now verify claims {\em (i)-(iv)}.
We first show that in the
formula (\ref{eq:delta-i}) defining $\delta_i$, the denominator is
nonnegative and $\delta_i>1$. 
Note that 
\begin{equation}
r_1(i)+r_2(i)-r_3(i)-r_4(i)-b_i=e_i\ge R_i\ge r_1(i).\label{eq:r1-4}
\end{equation}
The equality is by the definition of the four terms; 
the first inequality is required by $\Delta$-feasibility, and the
second since the definition of $R_i$ includes all terms in $r_1(i)$.
This shows that the denominator is $r_2(i)-r_4(i)-b_i\ge r_3(i)\ge 0$.
The inequality $\delta_i>1$ then follows by the equality in \eqref{eq:r1-4} and the assumption
$4(d_i+2)\Delta\mu_i> e_i$.

For {\em (i)}, 
the above argument gives $\alpha_1>1$. It is easy to see that $\alpha_2>1$,
and hence $\alpha>1$ follows.
For the second part, let us analyze the $\alpha=\infty$ case. First,
we show that $T=\emptyset$ must hold. For a contradiction,
assume $T\neq\emptyset$. We have $\sink\in V\setminus T$ is assumed, and  every $j\in V-\sink$ is connected by
an arc to $\sink$ by (\ref{cond:root}). Therefore the set of arcs defining $\alpha_2$ is
always nonempty, showing that $\alpha$ must be finite.

We thus have $T=\emptyset$. Since
$\alpha_1=\infty$, we must have $\delta_i=\infty$ for every $i\in V-\sink$, that is,
the denominator in  (\ref{eq:delta-i}) is always 0. According to \eqref{eq:r1-4},
this is only possible if $r_3(i)=0$ for every $i\in V-\sink$. This means that 
 for every $ij\in F^\mu$, if $f_{ij}>0$,  then $i=\sink$ must hold.
As a further consequence of \eqref{eq:r1-4}, we have $e_i=R_i=r_1(i)$  for every $i\in V-\sink$.
Combining these two, for every $i\in V-\sink$ we obtain
 $e_i=r_1(i)=\gamma_{ti}f_{ti}$ if $ti\in F^\mu$, and $e_i=r_1(i)=0$
 if $ti\notin F^\mu$.
After the algorithm sets the value of all these arcs to 0, $\mu$
becomes a conservative labeling, and $e_i=0$ for all $i\in V-\sink$, 
yielding primal and dual optimality according to Theorem~\ref{thm:genflow-opt}(i).

Let us now prove claim {\em (ii)}. The flow on the arcs
incident to $i$ is divided by $\alpha$ on all arcs in $F_1(i)$ and $F_3(i)$,
and left unchanged on arcs in $F_2(i)$ and $F_4(i)$. Therefore,
${e'_i}^{\mu'_i}\le \ketd{i}\Delta'$ follows
whenever $\alpha\le \delta_i$. 
The claims on nodes/arcs with equalities in {\em (ii)} and {\em(iv)}
are straightforward. 
Finally,  {\em (iii)} follows since if $i\in T$, then the flow is
unchanged on outgoing arcs and on arcs incoming from $T$, but decreases on arcs incoming from
$V\setminus T$.
\end{proof}

\subsection{Bounding the number of iterations}\label{sec:iter-weak}

Let $\Delta\itr{\tau}$ denote the value of
the scaling factor at the beginning of the $\tau$'th iteration; 
clearly, $\Delta^{(1)}\ge \Delta\itr{2}\ge\ldots\ge\Delta\itr{\tau}$.
Let
$f^{(\tau)},\mu^{(\tau)},e^{(\tau)}$ and $T\itr{\tau}$
denote the respective vectors and set $T$ at the beginning of iteration $\tau$.

Let us classify the iterations into three categories. The iteration $\theta$ is
{\em shrinking}, if  $T\itr{\theta}\setminus T\itr{\theta+1}\neq
\emptyset$.
Every iteration with a path augmentation is shrinking, since $T$ is
reset to $T_0$, although it contained some other nodes, in particular,
the endpoint $q$ of the path previously.
The other type of shrinking iteration is when
\estep{} is
performed, and for some $i\in T_0$, the value of  $e_i^\mu$ is decreased 
below $\egyd{i}\Delta$. 

The iteration $\theta$ is {\em expanding}, if $T\itr{\theta}\subsetneq
T\itr{\theta+1}$. This can either happen if the iteration only consists of
extending $T$ by adding a new node reachable by a tight arc in the
$\Delta$-fat graph, or if $T_0$ is extended in \estep{}, and no node
is removed from $T_0$. An iteration
that
is neither shrinking nor expanding is called {\em neutral}.
Note that in a neutral iteration we must perform \estep{}, and further
we must have  $T\itr{\theta}=T\itr{\theta+1}$. We claim that the iteration following the neutral iteration $\theta$ must be either expanding or shrinking. 
Indeed, if $T\itr{\theta+1}\cap N\itr{\theta+1}\neq\emptyset$, then it will be
shrinking. Otherwise, Lemma~\ref{lem:elementary}(iv) guarantees that
it must be expanding. 
The main goal of this section is to prove the following lemma.
\begin{lemma}\label{lem:iteration-bound}
For the starting value $\Delta\itr{1}=\bar\Delta$ and arbitrary integer $\tau\ge 1$,
we have 
\[
\tau\le 26mn\log_2\frac{\bar\Delta}{\Delta\itr{\tau+1}}.
\]
Further, the total number of shrinking iterations among the first $\tau$ is at most
\[
13m\log_2\frac{\bar\Delta}{\Delta\itr{\tau+1}}.
\]
\end{lemma}

An important quantity in our analysis will be
\[
\beta_i:=\frac{e_i}{\Delta\mu_i};
\]
let $\beta_i\itr{\tau}$ denote the corresponding value at the
beginning of iteration $\tau$.
Let $\alpha\itr\tau$ denote the value of $\alpha$ in iteration $\tau$ if
 the subroutine \textsc{Elementary Step} is called, and let
 $\alpha\itr\tau=1$ otherwise.
Note that the value of the scaling factor only changes in the
subroutine \textsc{Elementary Step}.
Therefore 
\[
\frac{\bar\Delta}{\Delta\itr{\tau+1}}=\prod_{\theta\in[1,\tau]}\alpha\itr{\theta} \quad \forall \tau\in \Z,\ \tau>1.
\]

\begin{lemma}\label{lem:sigma-bound}
During the first $\tau$ iterations, a node $i$ may enter the set $T_0$ altogether at most
$\log_2 \frac{\bar\Delta}{\Delta\itr{\tau+1}}$ 
times.
\end{lemma}
Before proving the lemma, let us show how it can be used to bound the
number of iterations.

\begin{proof}[Proof of Lemma~\ref{lem:iteration-bound}]
Let us consider the potential
\begin{equation}
\Psi:=\sum_{i\in T_0} \lfloor \beta_i-\negd{i}\rfloor.\label{def:psi}
\end{equation}
Initially, $T_0=\emptyset$ and therefore $\Psi=0$. Note that every
term is positive in every step of the algorithm, since nodes with 
$\beta_i<\egyd{i}$ are immediately removed from $T_0$.
The subroutine \estep{} may only decrease the value of $\Psi$:
Lemma~\ref{lem:elementary}(iii) guarantees that if $i\in T_0$, then $\beta_i$ may only
decrease during the subroutine, since  $e_i'\le e_i$ and $\Delta'\mu'_i=\Delta\mu_i$.

Every shrinking iteration must decrease $\Psi$ by at least one. Indeed, a path
augmentation decreases $e_p$ by $\Delta\mu_p$ for the starting node
$p$, which decreases $\lfloor \beta_p-\negd{p}\rfloor$ by one. 
No other $\beta_i$ value is modified for $i\in T_0$. 
Next, consider the case when a shrinking iteration removes some nodes $i$ from $T_0$ after
performing \estep{} because of $\beta_i<\egyd{i}$. In the previous
iteration, we must have had $\beta_i\ge \egyd{i}$ for such nodes,  hence $\Psi$
decreases by at least 1.

When a node $i$ enters $T_0$, then it increases $\Psi$ by
$(3d_i+7)$. 
Assume that  the node $i$ enters $T_0$ altogether $\lambda_i$
times between iterations $1$ and $\tau$. Then 
 Lemma~\ref{lem:sigma-bound}
gives 
$\lambda_i\le  \log_2 \frac{\bar\Delta}{\Delta\itr{\tau+1}}$.
Therefore the total increase in the $\Psi$ value between iterations $1$
and $\tau$ is bounded by 
\[
\sum_{i\in V-\sink}(3d_i+7)\lambda_i\le 
\sum_{i\in V-\sink}(3d_i+7)\log_2
\frac{\bar\Delta}{\Delta\itr{\tau+1}}\le (6m+7n) \log_2
\frac{\bar\Delta}{\Delta\itr{\tau+1}}\le 13m \log_2
\frac{\bar\Delta}{\Delta\itr{\tau+1}}
\] 
This bounds the number of shrinking iterations (recall the assumption
$n\le m$). Between two subsequent
shrinking iterations, all phases are expanding or neutral. Every expanding 
iteration increases $T$, and every neutral iteration is followed by a
shrinking or an
expanding iteration. Therefore the total number of iterations between
two subsequent shrinking iterations is $\le 2n$, giving an overall
bound
\[
26mn\log_2
\frac{\bar\Delta}{\Delta\itr{\tau+1}}
\] 
on the number of iterations.
\end{proof}

The proof of Lemma~\ref{lem:sigma-bound} is based on the following
simple claim.

\begin{claim}\label{cl:e-inc}
Let $\beta'_i$ denote the new value of $\beta_i$
after performing 
the subroutine \textsc{Elementary Step}$(T,T_0,f,\mu,\Delta)$, that computes
the value $\alpha$.
For every node $i\in V-\sink$, we have
\[
\beta'_i \le \alpha^2\max\left\{\beta_i,\ind{i}\right\}.
\]
\end{claim}
\begin{proof}
Let $\Delta$ and $\Delta'=\Delta/\alpha$ denote the scaling factor
before and after performing the subroutine \textsc{Elementary
  Step}$(T,T_0,f,\mu,\Delta)$.
If $i\in T$, then $e'_i\le e_i$ by Lemma~\ref{lem:elementary}(iii) and $\Delta'\mu'_i=\Delta\mu_i$, and hence $\beta'_i\le \beta_i$, implying the claim.
Assume therefore that $i\in V\setminus T$.
We have
$f'\le f$, and the flow changes on arcs entering $i$ may only decrease
$e_i$. Recall that $F_3(i)$ denotes the set of outgoing arcs
 $ij$ where $f'_{ij}<f_{ij}$. Note that $f_{ij}\le \Delta\mu_i$ on
 every such arc.
We get the upper bound
\[
e'_i\le e_i+\sum_{j:ij\in F_3(i)} (1-1/\alpha)f_{ij}\le
e_i+(1-1/\alpha)|F_3(i)|\Delta\mu_i\le e_i+(\alpha-1)|F_3(i)|\Delta\mu_i.
\]
In the last inequality, we used $1-1/\alpha\le \alpha-1$, which is
true for every $\alpha>0$.
Using further that $\Delta'\mu'_i=\Delta\mu_i/\alpha$, we get
\begin{align*}
\beta'_i=
\frac{e'_i}{\Delta'\mu'_i}\le
\frac{\alpha(e_i+(\alpha-1)|F_3(i)|\Delta\mu_i)}{\Delta\mu_i}=
\alpha\frac{e_i}{\Delta\mu_i}+(\alpha-1)|F_3(i)|\le \\
\alpha\beta_i+(\alpha-1)\ind{i}\le (2\alpha-1)\max\{\beta_i,\ind{i}\}\le \alpha^2\max\{\beta_i,\ind{i}\},
\end{align*}
completing the proof.
\end{proof}

\begin{proof}[Proof of Lemma~\ref{lem:sigma-bound}]
Let $\tau_1< \tau_2< \ldots <
\tau_\lambda\le \tau$
denote the iterations when $i$ enters $T_0$ up to iteration $\tau$.
 This means that
 $\beta_{i}\itr{\tau_\ell+1}=\ketd{i}$ for
 $1\le \ell\le \lambda$.

For $1\le\ell\le\lambda$, let us define $\tau'_\ell$ to be the largest
value $\tau'_\ell\le \tau_\ell$ such that
$\beta_{i}\itr{\tau'_\ell}< \egyd{i}$. 
Note that these values must exist and satisfy
$\tau_{\ell-1}<\tau'_\ell\le \tau_\ell$ for $\ell>1$.
Indeed, for $\ell=1$, we assumed that at the beginning of the
algorithm
$\beta_{i}\itr{1}<\egyd{i}$. For $\ell>1$, note that $i$
must leave $T_0$ in some iteration $\theta$ between $\tau_{\ell-1}$ and
$\tau_{\ell}$, and this can happen only if $\beta_i^{(\theta)}<\egyd{i}$.

In iteration $\tau'_\ell$, we have $i\notin T_0$, since once the
excess $e_i$ drops below 
$\egyd{i}\Delta\mu_i$, the node $i$ is immediately removed from $T_0$.
By definition, $i$ will be added to $T_0$ in iteration $\tau_\ell$.

The $e_i$ values may change in two ways between iterations
$\tau'_\ell$ and $\tau_\ell$: either during
a path augmentation or in the subroutine \estep. We claim that no path
augmentation changes $e_i$ in the iterations  $\tau'_\ell\le \theta\le
\tau_\ell$.
Indeed, the only values that change are at the starting point $p$ and
endpoints $q$ of
the tight path $P$. We cannot have $i=p$ as $i\notin T_0$ during these
iterations. Assume now $i=q$ is the endpoint; therefore
$e_i\itr{\theta}<\negd{i}\Delta\itr{\theta}\mu_i\itr{\theta}$.
This clearly cannot be the case for $\tau'_\ell<\theta\le \tau_\ell$
by the maximal
choice of $\tau'_\ell$. Let us consider the case $\theta=\tau'_\ell$. The path augmentation
terminating in $i=q$ increases $e_i\itr{\tau'_\ell}$ by $\Delta\itr{\tau'_\ell}\mu_i\itr{\tau'_\ell}$.
However, we had $e_i\itr{\tau'_\ell}<\negd{i}\Delta\itr{\tau'_\ell}\mu_i\itr{\tau'_\ell}$, and
therefore
\[
 e_i\itr{\tau'_\ell+1}=e_i\itr{\tau'_\ell}+\Delta\itr{\tau'_\ell}\mu_i\itr{\tau'_\ell}<\egyd{i}\Delta\itr{\tau'_\ell+1}\mu_i\itr{\tau'_\ell+1},
\]
again a contradiction to the choice of $\tau'_\ell$. (Note that if a path augmentation
is done in iteration $\tau'_\ell$, then the values of $\Delta$ and $\mu$ do
not change).

Hence all changes in the value of $e_i$ are due to modifications in
\estep. 
Consequently,
\begin{equation}\label{eq:four}
4=\frac{\ketd{i}}{\egyd{i}}< \frac{\beta\itr{\tau_\ell+1}}{\max\{\beta\itr{\tau'_\ell},\ind{i}\}}\le
\frac{\beta\itr{\tau'_\ell+1}}{\max\{\beta\itr{\tau'_\ell},\ind{i}\}}\prod_{\theta\in[\tau'_\ell+1,\tau_\ell]}\frac{\beta\itr{\theta+1}}{\beta\itr{\theta}}
\end{equation}
For $\theta\in[\tau'_\ell+1,\tau_\ell]$, we assumed $\beta\itr{\theta}>\ind{i}$, and hence Claim~\ref{cl:e-inc}
gives that $\frac{\beta\itr{\theta+1}}{\beta\itr{\theta}}\le
\left(\alpha\itr{\theta}\right)^2$. The same claim bounds the first
term by
 $\le \left(\alpha\itr{\tau'_\ell}\right)^2$. Hence we get
\[
4\le \left(\prod_{\theta\in [\tau'_\ell,\tau_\ell]} \alpha\itr{\theta}\right)^2.
\]
Adding the logarithms of these inequalities  for all $\ell=1,\ldots,\lambda$, we obtain
\[
\lambda\le \sum_{\theta\in [1,\tau]} \log_2\alpha\itr{\theta}=\log_2 \frac{\bar\Delta}{\Delta\itr{\tau+1}},
\]
completing the proof.
\end{proof}

\subsection{The termination of the algorithm}\label{sec:final-opt}
The algorithm either terminates in \estep{} or by the final subroutine 
\textsc{Tight-flow}$(V,\mu)$. Optimality for the first case was
already proved in Lemma~\ref{lem:elementary}(i). The next claim
addresses the second case.
\begin{lemma}\label{lem:final-opt}
The final $f'$ and $\mu$ returned by the subroutine
\textsc{Tight-flow}$(V,\mu)$
are a primal and a dual optimal  solution to (\ref{primal}) and (\ref{dual}), respectively.
\end{lemma}
\begin{proof}
We show that the flow problem in \textsc{Tight-flow}$(V,\mu)$ is feasible and $Ex^\mu(f')<1/\bar B^3$. Then optimality follows by Theorem~\ref{thm:tight-flow}(iii).
At the termination of the While iterations of the
algorithm \textsc{Continuous Scaling}, we have 
\[
Ex^\mu(f)=\sum_{i\in V-\sink}e_i^\mu\le 4\Delta\sum_{i\in
  V}\egyd{i}=(8m+8n)\Delta.
\]
Let us define $\tilde f$ by $\tilde f_{ij}=0$ if $ij\in F^\mu$ and
$\tilde f_{ij}=f_{ij}$ otherwise. 
By Lemma~\ref{lem:make-conservative},
\[
Ex^\mu(\tilde f)<Ex^\mu(f)+|F^\mu|\Delta\le (9m+8n)\Delta< 1/\bar B^3,
\]
since $\Delta<\threshold$ at the termination.
The proof is complete by verifying the feasibility of the flow problem and showing that $Ex^\mu(f')\le Ex^\mu(\tilde f)$.

Let us define the feasible solution $\tilde x$ to the flow problem in
\textsc{Tight Flow} as follows. We use the notation introduced in the
description of the subroutine in Section~\ref{sec:tight-flow}.
Let $\tilde x_{ij}:=\tilde f^\mu_{ij}$ for $ij\in E$. Further, for $i\in
V-\sink$, let us set 
$\tilde x_{si}:=\sum_{j:ij\in E}\tilde f^\mu_{ij}-\sum_{j:ji\in  E} \tilde f^\mu_{ji}$. The conservativeness of $\tilde f$ implies that 
$\tilde x_{si}\le -b^\mu_i=u_{si}$. Therefore
$\tilde x$ is a feasible solution to the flow problem.
The value of this flow $\tilde x$ (i.e. the sum of the flow on the
arcs leaving $s$) is
\[
\sum_{i\in \tilde V-\sink}\tilde x_{si}=-\sum_{i\in 
  V-\sink}(b_i^\mu+e_i^\mu(\tilde f))=-Ex^\mu(\tilde f)-\sum_{i\in 
  V-\sink}b_i^\mu.
\]
Similarly, the value of the flow $x$ found by \textsc{Tight Flow}   is $-Ex^\mu(f')-\sum_{i\in 
  V-\sink}b_i^\mu$. Since $x$ is maximal, it follows that 
$Ex^\mu(f')\le Ex^\mu(\tilde f)$.
\end{proof}

\subsection{Running time analysis}

\begin{proof}[Proof of Theorem~\ref{thm:weak-running}]
The starting value of the scaling factor is $\bar\Delta\le \Deltabound$ by Lemma~\ref{lem:init}, and
we terminate once $\Delta\itr{\tau+1}<\threshold$. Therefore
$\log\frac{\bar \Delta}{\Delta\itr{\tau+1}}\in O(\log \bar B)$ (we may assume $\log \bar B$ is larger than $m$).
According to
Lemma~\ref{lem:iteration-bound}, the number of iterations of the
algorithm is $O(mn\log \bar B)$, out of them $O(m\log \bar B)$
shrinking ones.
We have to execute two maximum flow computations, that can be done in
$O(nm)$ time using the recent algorithm by Orlin \cite{Orlin13}. The
initial cycle canceling subroutine can be executed in time
$O(m^2n\log^2 n)$, see Radzik
\cite{Radzik93}.
The proof is complete by showing that the part of the algorithm
between two shrinking iterations can be implemented in $O(m+n\log n)$ time.
 
We implement all these iterations together via a Dijkstra-type
algorithm, using the Fibonacci-heap data structure \cite{Fredman87},
see also \cite[Chapter 4.7]{amo}. The precise details are given in
Section~\ref{sec:rounding}, see Figure~\ref{code:block}; here we
outline the main ideas only. 
Each label is modified only once, at the beginning of the
subsequent shrinking iteration; for every $i$, it is sufficient to
record the value of $\alpha$ at the moment when $i$ enters $T$. We
have to modify the $f_{ij}$ values accordingly.
We maintain a heap with elements $i\in V\setminus T$, with five keys associated
to each of them. The main key for $i\in V\setminus T$  corresponds
to the minimum of the  $1/\gamma_{ji}^\mu$'s for $j\in T$, and of
$\delta_i$.
The four auxiliary keys store the flow values $r_1(i),\ldots,r_4(i)$, as in the definition 
(\ref{eq:delta-i}) of $\delta_i$. 
We choose the next $i$ who enters $T$ with the minimal main key. If
the minimal key corresponds to the $\delta_i$ value, then $i$ enters
both $T$ and $T_0$; otherwise, it enters only $T$.
We remove $i$ from the heap, and update the keys on the adjacent
nodes. We maintain another heap structure on $T$ to identify events
when for a node  $i\in T_0$, $e_i^\mu<\egyd{i}\Delta$ happens, or when
a node in $T\setminus T_0$ enters $N$.

Overall, these modifications entail $O(m)$ key modifications only;
the keys can be initialized in total time $O(m)$. We
therefore obtain the running time $O(m+n\log n)$ as for Dijkstra's algorithm.
\end{proof}

\section{The strongly polynomial algorithm}\label{sec:strong}
The while loop of the algorithm \textsc{Enhanced Continuous Scaling}  proceeds very similarly to \textsc{Continuous Scaling}, with the addition of the special subroutine \textsc{Filtration}, described in Section~\ref{sec:filtration}.
However, the termination criterion is quite different.
As discussed in Section~\ref{sec:overview}, the goal is to
find a node $i\in V-t$ with $\frac{|b_i^\mu|}{\Delta}\ge 32mn$. There must be an abundant arc incident to such a node that we can contract and 
continue the algorithm in the smaller graph.  Section~\ref{sec:abundant} describes the abundant arcs and the contraction operation.

Let us now give some motivation for the algorithm; we focus on the sequence of iterations leading to the first abundant arc.
Consider the set 
\[
D:=\left\{i\in V-t: \frac{|b_i^\mu|}{\Delta}\ge \frac{1}{n}\right\}.
\]
Our aim is to  guarantee that most iterations when $\Delta$ is multiplied by $\alpha$ will multiply $\frac{|b_i^\mu|}{\Delta}$ by $\alpha$ for some $i\in D$. 
This will ensure that  $\frac{|b_i^\mu|}{\Delta}\ge  32mn$ happens within $O(nm\log n)$ number of steps.
Note that in the subroutine \estep$(T,f,\mu,\Delta)$, the $\frac{|b_i^\mu|}{\Delta}$ ratio is
multiplied by $\alpha$ for all nodes $i\in V\setminus T$ and remains
unchanged for $i\in T$. 

Therefore we modify the while loop of \textsc{Continuous Scaling} as follows.
If $(V\setminus T)\cap D\neq\emptyset$,  \estep$(T,f,\mu,\Delta)$ is performed identically.
If $(V\setminus T)\cap D= \emptyset$, then before \estep$(T,f,\mu,\Delta)$, the special
subroutine \textsc{Filtration}$(V\setminus T,f,\mu)$ is executed, performing the following changes.

The value of $f$ is set to 0 for every arc entering $T$, and
$f_{ij}$ is left unchanged for $i\in T$. The flow value on arcs  inside $E[V\setminus
T]$ is replaced by an entirely new flow $f'$ computed by \textsc{Tight Flow}$(V\setminus T,\mu)$.

An important part of the analysis is Theorem~\ref{thm:tight-flow}(ii), asserting that $e^\mu_i(f')\le n\max_{j\in (V\setminus T)-t}|b_j^\mu|$.
This will imply that either the set $D$
must be extended in the iteration following \textsc{Filtration}$(V\setminus T,f,\mu)$, or
there must be a shrinking one among the next two
iterations (Lemma~\ref{lem:fluvial}(ii)).
Note that once a node enters $D$, it stays there until the next contraction.

\subsection{Abundant arcs and contractions}\label{sec:abundant}
Given a $\Delta$-feasible pair $(f,\mu)$, we say that an arc $pq\in E$
is {\em abundant}, if  $f^\mu_{pq}\ge\arcbound$. The importance of abundant arcs is that they must be tight in all dual optimal
solutions. This is a corollary of the following theorem.
\begin{theorem}\label{thm:close-opt}
Let $(f,\mu)$ be a $\Delta$-feasible pair. Then there exists an optimal solution
$f^*$ such that 
\[
||f^\mu-{f^*}^\mu||_\infty \le Ex^\mu(f)+(|F^\mu|+1)\Delta.
\]
\end{theorem}
The standard proof using flow decompositions is given in
the Appendix; it can also be derived from Lemma 5 in
Radzik \cite{Radzik04}.
For the flow $f$ in an iteration with scaling factor $\Delta$, we
have $Ex^\mu(f)\le \sum_{i\in V-\sink}\ketd{i}\Delta< (8m+8n-8)\Delta\le (16m-8)\Delta$. Further,
$|F^\mu|\le m$. This gives the following corollary; the last part follows by primal-dual slackness conditions.
\begin{corollary}\label{cor:abundant}
Let $(f,\mu)$ be the $\Delta$-feasible pair during the algorithm. If for an arc $pq\in E$, $f^\mu_{pq}\ge \arcbound$, then $f^*_{pq}>0$ for some optimal solution $f^*$ to (\ref{primal}). Consequently, $\gamma_{pq}\mu^*_p=\mu^*_q$ for every optimal solution $\mu^*$ to (\ref{dual}).
\end{corollary}

Once  an
abundant arc $pq$ is identified in the \textsc{Enhanced Continuous Scaling} algorithm,
it is possible to reduce the problem by contracting $pq$.
Consider the problem instance $(V,E,\sink,b,\gamma)$.
The {\em contraction of the arc $pq$} returns a problem instance
$(V',E',\sink',b',\gamma')$ with $\sink':=\sink$, as follows.

{\bf Case I: $p\neq \sink$}.
Let $V'=V\setminus \{p\}$, and add an arc
$ij\in E'$ if $ij\in E$ and $i,j\neq p$. For every arc $ip\in E$, add
an arc $iq\in E'$, and for every arc $pi\in E$, $i\neq q$, add an arc $qi\in E'$. Set
the gain factors as $\gamma'_{ij}:=\gamma_{ij}$ if $i,j\neq p$,
$\gamma'_{iq}:=\gamma_{ip}\gamma_{pq}$ and
$\gamma'_{qi}:=\gamma_{pi}/\gamma_{pq}$. Let us set
$b'_i:=b_i$ if $i\neq q$, and
$b'_q:=b_q+\gamma_{pq}b_p$. 

{\bf Case II: $p=\sink$}.
Let $V'=V\setminus \{q\}$, and add an arc
$ij\in E'$ if $ij\in E$ and $i,j\neq q$. For every arc $iq\in E$,
$i\neq p$, add
an arc $ip\in E'$, and for every arc $qi\in E$, add an arc $pi\in E'$. Set
the gain factors as $\gamma'_{ij}:=\gamma_{ij}$ if $i,j\neq p$,
$\gamma'_{ip}:=\gamma_{iq}/\gamma_{pq}$ and
$\gamma'_{pi}:=\gamma_{qi}\gamma_{pq}$. Let us set
$b'_i:=b_i$ if $i\neq p$, and
$b'_p:=b_p+b_q/\gamma_{pq}$. 

In both cases, if parallel arcs are created,
keep only one that maximizes the $\gamma'$ value. 
Let $s:=q$ in the first and $s:=p$ in the second case.
If a loop incident to $s$ is created (corresponding to a $qp$ arc),
remove it.

Assume further we are given a generalized flow $f$ and a labeling
$\mu$ with $\gamma_{pq}^\mu=1$ in
the instance. We define the image labels $\mu'$, by simply setting
$\mu'_i=\mu_i$ for all $i\in V'$ in both cases.
Note that we will have ${b'_s}^{\mu'}=b_p^\mu+b_q^\mu$ in both cases.

As for the generalized flow, let $f'_{ij}:=f_{ij}$ whenever $i,j\neq
s$. For every $i\in V'\setminus \{s\}$, we let
$f'_{is}:=f_{ip}+f_{iq}$. Further, in Case I, we let
$f'_{si}:=\gamma_{pq}f_{pi} +f_{qi}$, whereas in Case II, we let 
$f'_{si}:=f_{pi} +f_{qi}/\gamma_{pq}$. If one of these arcs is not in
$E$, then we substitute the corresponding value by 0. 
Recall that in the construction, we keep the larger gain factor 
from two parallel incoming or outgoing arcs.
%

The above transformation of an instance, generalized flow and labels will be
executed by the subroutine \textsc{Contract}$(pq)$. Note that if the
original instance satisfies (\ref{cond:root}), (\ref{cond:init}), and (\ref{cond:bounded}),
then these also hold for the contracted instance; the contracted image of
the initial feasible solution $\bar f$ is feasible for the contracted instance.

Let us also describe the reverse operation,
\textsc{Reverse}$(pq)$, that transforms a dual solution on the
contracted instance to a dual solution in the original one.
 Assume $\mu'$ is a dual solution in
the graph obtained by the contraction of $pq$. 
Let us set $\mu_i:=\mu'_i$ for all $i\in V-s$. In the
first case ($p\neq \sink$, $s=q$), let us set $\mu_p:=
\mu_q'/\gamma_{pq}$,
whereas in the second case ($p=\sink$, $s=p$), let us set $\mu_q:=\mu'_p\gamma_{pq}=\gamma_{pq}$.

\subsection{The Filtration subroutine}\label{sec:filtration}
A typical iteration of the \textsc{Enhanced Continuous Scaling} algorithm (Figure~\ref{code:strong})  will be the same as in  \textsc{Continuous Scaling}, with adding one additional subroutine, \textsc{Filtration}$(V\setminus T,f,\mu)$ before performing \estep $(T,T_0,f,\mu,\Delta)$.
 This subroutine is executed if
$|b^\mu_i|<\Delta/(16^\ka n)$ holds for all $i\in  (V\setminus
T)-\sink$, where $\ka$ is the number of arcs contracted so far,
initially $\ka=0$.

\textsc{Filtration}$(V\setminus T,f,\mu)$  (Figure~\ref{code:filtration}) performs the  subroutine
\textsc{Tight Flow}$(V\setminus T,\mu)$, as described in
Section~\ref{sec:tight-flow}. This replaces $f$ by an entirely new
flow $f'$ on the arcs in $E[V\setminus T]$. We further set $f_{ij}=0$
on all arcs entering $T$, and keep the original $f$ value on all other
arcs (that is, arcs in $E[T]\cup E[T,V\setminus T]$).

\begin{figure}[htb]
\begin{center}
\fbox{\parbox{\textwidth}{
\begin{tabbing}
xxxxx \= xxx \= xxx \= xxx \= xxxxxxxxxxx \= \kill
\> \textbf{Subroutine} \textsc{Filtration}$(V\setminus T,f,\mu)$\\
\> $f'\leftarrow$ \textsc{Tight Flow}$(V\setminus T,\mu)$ ;\\
\> \textbf{for} $ij\in E$ \textbf{do} \\
\> \> \textbf{if} $ij\in E[V\setminus T]$ \textbf{then}
$f_{ij}\leftarrow f'_{ij}$ ;\\
\> \> \textbf{if} $ij\in E[V\setminus T,T]$ \textbf{then}
$f_{ij}\leftarrow 0$ ;
\end{tabbing}
}}
\caption{The Filtration subroutine}\label{code:filtration}
\end{center}
\end{figure}

\subsection{The Enhanced Continuous Scaling Algorithm}
We are ready to describe our strongly polynomial algorithm, shown on
Figure~\ref{code:strong}.
The algorithm consists of iterations similar to  \textsc{Continuous Scaling}, with the addition of the above described \textsc{Filtration} subroutine.
This subroutine might decrease $e_i^\mu$ values below $\egyd{i}\Delta$
for some $i\in T_0$; also, $e_i^\mu<\negd{i}\Delta$ might happen for some $i\in T$, that is, $i$ is added to the set $N\cap T$. If either of these events happen, we proceed to the next iteration without performing the subroutine \estep$(T,T_0,f,\mu,\Delta)$.
Further, if there are nodes $i\in T_0$ where the $e_i^\mu$ values drop below $\egyd{i}\Delta$, then we remove all such nodes from $T_0$, and reset $T=T_0$.

\begin{figure}[!htb]
\begin{center}
\fbox{\parbox{\textwidth}{
\begin{tabbing}
xxxxx \= xxx \= xxx \= xxx\= xxx \= xxx \= xxxxxxxxxxx \= \kill
\> \textbf{Algorithm} \textsc{Enhanced Continuous Scaling}\\
\> \textsc{Initialize}$(V,E,b,\gamma,\bar f)$ ; \\
\> $T_0\leftarrow \emptyset$ ; $T\leftarrow \emptyset$ ;\\
\> $\ka\leftarrow 0$ ;\\
\> \textsc{While} $|V|>1$ \textsc{do}\\
\> \> $N\leftarrow\{t\}\cup \{i\in V-\sink: e_i^\mu<\negd{i}\Delta\}$ ;\\
\>  \> \textbf{if} $N\cap T\neq \emptyset$ \textbf{then}\\
\> \> \>  \textbf{pick} $p\in T_0$, $q\in N\cap T$ connected by a
tight path $P$ in $E^\mu_f(\Delta)$ ;\\
\> \> \> \textbf{send} $\Delta$ units of relabeled flow from $p$ to
$q$ along $P$ ;\\
\> \> \> \textbf{if} $e_p^\mu<\egyd{p}\Delta$ \textbf{then} $T_0\leftarrow
T_0\setminus \{p\}$ ;\\
\> \> \> $T\leftarrow T_0$ ;\\
\> \> \textbf{else}\\
\> \> \> \textbf{if} $\exists ij\in E_f^\mu(\Delta)$,
$\gamma_{ij}^\mu=1$, $i\in T$, $j\in V\setminus T$ \textbf{then}
$T\leftarrow T\cup \{j\}$ ;\\
\> \> \> \textbf{else}\\
\> \> \> \> \textbf{if} $\left(\forall i\in (V\setminus T)-\sink: |b^\mu_i|<\frac{\Delta}{16^\ka n}\right)$ \textbf{then} \textsc{Filtration}$(V\setminus T,f,\mu)$ ;\\
\> \> \> \> \textbf{if} ($e_i^\mu\ge \egyd{i}\Delta$ for all $i\in T_0$) {\em and} ($e_i^\mu\ge \negd{i}\Delta$  for all $i\in T$)\\
\> \> \> \> \> \textbf{then} \estep $(T,T_0,f,\mu,\Delta)$ ;\\
\> \> \> \>\textbf{elseif} $\exists i\in T_0: e_i^\mu< \egyd{i}\Delta$  \textbf{then}\\
\> \> \> \>\>  $T_0\leftarrow T_0\setminus \{i:e_i^\mu<\egyd{i}\Delta\}$ ;\\
\> \> \> \>\>  $T\leftarrow T_0$ ;\\
\> \> \textbf{while} $\exists\ pq\in E$: $f_{pq}^\mu\ge \arcbound$
\textbf{do} \\
\> \> \> \textbf{for all} $ij\in E:\ \gamma_{ij}^\mu<1$ \textbf{do}
$f_{ij}\leftarrow 0$ ;\\
\> \> \> \textsc{Contract}($pq$) ;\\
\> \> \> $\Delta\leftarrow 16\Delta$ ;\\
\> \> \> $\ka\leftarrow \ka+1$ ;\\
\> \> \> $T_0\leftarrow \emptyset$ ; $T\leftarrow \emptyset$ ;\\
\> \textsc{Expand-to-Original}$(\mu)$ ;
\end{tabbing}
}}
\caption{Description of the strongly polynomial algorithm}\label{code:strong}
\end{center}
\end{figure}

The termination criterion is not on the value of $\Delta$, but on the
size of the graph: we terminate once it is reduced to a single node.
The main progress is done when an abundant arc $pq$ appears: in this
case, we first set the flow value on every non-tight arc to 0, and then reduce the number of nodes by one using the above described subroutine \textsc{Contract}$(pq)$.
Further, the value of the scaling factor $\Delta$ is multiplied
by 16, and the counter $k$ is increased by one.
 The sets $T_0$ and $T$
are reset to $\emptyset$.
A sequence of such contractions is performed until all abundant arcs are contracted.
The iterations between two phases where contractions are performed (and those up to the first contraction) will be referred to as a {\em major cycle} of the algorithm.
In the description and the analysis, $n$ and $m$ will always
refer to the size of the original instance and not the actual
contracted one.

At termination, the subroutine \textsc{Expand-to-Original} finds 
primal and dual optimal solutions in the original graph. This is done
by first expanding all contracted arcs $pq$ by the subroutine
\textsc{Reverse}$(pq)$, taking these arcs in the reverse order of
their contraction. Hence we obtain a dual optimal solution $\mu^*$ in
the original graph (see Lemma~\ref{lem:contract-opt}). Finally, the
subroutine \textsc{Tight-flow}$(V,\mu^*)$ obtains a primal optimal
solution, as guaranteed by Theorem~\ref{thm:tight-flow}(i).

\begin{theorem}\label{thm:strong-running}
The algorithm \textsc{Enhanced Continuous Scaling} 
finds an optimal solution
for the uncapacitated formulation (\ref{primal}) in
running time $O(n^3m^2)$ elementary arithmetic operations and comparisons.
\end{theorem}

To get a truly strongly polynomial algorithm, we also need to
guarantee that the size of the numbers during the computations remain
polynomially bounded. We shall modify the algorithm in
Section~\ref{sec:rounding} by incorporating additional rounding steps
to achieve that.

We remark that the algorithm can be simplified by terminating once
the first abundant arc is found, and restarting from scratch on the
contracted graph. This
would give a running time bound $O(n^3m^2\log n)$: hence, we are able
to save a factor $\log n$ by continuing with the contracted image of
the current flow instead of a fresh start.
\section{Analysis of the strongly polynomial algorithm}\label{sec:strong-analysis}
Many properties of the \textsc{Continuous Scaling} algorithm derived in 
Section~\ref{sec:anal-weak} remain valid. In particular,
Lemmas~\ref{lem:init} and \ref{lem:elementary}, and
Claims~\ref{cl:path-aug} and \ref{cl:e-inc} are applicable with
repeating the proofs verbatim. The argument bounding the number of iterations will be an extension of the one in Section~\ref{sec:iter-weak}.

\subsection{Properties of dual solutions}\label{sec:dual-prop}
Let us first verify that expanding the dual optimal
solution of the contracted instance results in a valid dual optimal
solution of the original instance.

\begin{lemma}\label{lem:contract-opt}
Assume that $pq\in E$ satisfies  $\gamma_{pq}\mu^*_p=\mu^*_q$ for every optimal solution $\mu^*$ to (\ref{dual}) for the problem instance $(V,E,\sink,b,\gamma)$.
Let $\mu'$ be an optimal solution to (\ref{dual}) to the contracted instance 
$(V',E',\sink',b',\gamma')$ obtained by the subroutine
\textsc{Contract}$(pq)$. If $p\neq \sink$, then let $\mu_i:=\mu'_i$
for every
$i\in V-p$ and let $\mu_p:=\mu'_q/\gamma_{pq}$.
 If $p=\sink$, then let $\mu_i:=\mu'_i$ for every $i\in V-q$ and let $\mu_q=\gamma_{pq}$. 
Then $\mu$ is an optimal solution to (\ref{dual}) in the original instance $(V,E,\sink,b,\gamma)$.
\end{lemma}
\begin{proof}
We give the proof to the $p\neq t$ case only; the other case follows similarly.
First, let us verify that $\mu$ is a feasible solution  to (\ref{dual}). 
It is straightforward
that $\mu_\sink=1$ and $\mu_i>0$ if $i\in V-\sink$.
Also, $\gamma_{ij}^\mu\le 1$ is straightforward if $i,j\neq q$, and
$\gamma_{pq}^\mu= 1$.  For an arc $ip\in E$, let $iq\in E'$ denote its image. Then  $\gamma'_{iq}\frac{\mu'_i}{\mu'_q}\le 1$, which can be written as $\gamma_{ip}\gamma_{pq}\frac{\mu_i}{\mu_p\gamma_{pq}}\le 1$, giving $\gamma_{ip}^\mu\le 1$.
One can verify $\gamma_{pi}^\mu\le 1$ for every $pi\in E$ analogously.

Assume for a contradiction that $\mu$ is not optimal to (\ref{dual}): there exists an optimal solution $\mu^*$ with $\sum_{i\in V}b_i^{\mu^*}>\sum_{i\in V}b_i^{\mu}$.
By our assumption, $\gamma_{pq}\mu^*_p=\mu^*_q$ must hold.
Consider the restriction of $\mu^*$ to $V'=V\setminus\{p\}$; it is
easy to check that it is feasible to (\ref{dual}) in the contracted
instance. Using $b'_p=b_p+\gamma_{pq}b_q$, and thus
${b'_s}^{\mu^*}=b_p^{\mu^*}+b_q^{\mu^*}$, and ${b'_s}^{\mu'}=b_p^{\mu}+b_q^{\mu}$,
we obtain a contradiction by
\[
\sum_{i\in V}b_i^{\mu'}<\sum_{i\in V}b_i^{\mu^*}=\sum_{i\in
  V'}{b'_i}^{\mu^*}\le \sum_{i\in V'}{b'_i}^{\mu'}=\sum_{i\in V}{b_i}^{\mu'}.
\]
\end{proof}

Our next claim justifies that the feasibility properties are
maintained during the algorithm.
\begin{claim}\label{cl:contract-feas}
Let $\Delta':=16\Delta$, and let $f'$ and $\mu'$ denote the flow and
labels after contracting the abundant arc $pq$. Then $\mu'$ is a
conservative labeling for $f'$, with
$e_i^{\mu'}(f')< \egyd{i}\Delta'$ for all $i\in V-\sink$.
\end{claim}
\begin{proof}
Before the contraction, the flow on every non-tight arcs is set to 0;
this increases $e^\mu_i$ on every node by at most $d_i\Delta$.
Let $s=p$ or $s=q$ denote the contracted node. 
It is straightforward by the properties of the contraction that if
$e'$ is the image of the arc $e$, then
$\gamma_{e'}^{\mu'}=\gamma_{e}^\mu$. Since $\mu$ is conservative for
$f$ before the contraction, it follows that $\mu'$ is conservative for $f'$.
%

Consider a node $i\neq s$. Setting the flow values on non-tight
arcs to 0 increased $e^\mu_i$ by at most $d_i\Delta$, and
$e_i^{\mu'}(f')=e_i^\mu(f)$, and hence $e_i^{\mu'}(f')\le 
(5\ind{i}+8)\Delta<\egyd{i}\Delta'$. 
Let us now consider the contracted node $s$.  There is nothing to prove about $e_s^{\mu'}(f')$ if $s=\sink$, hence we may assume $s\neq\sink$.
Before the contraction, we had 
$e_p^{\mu}(f)\le (5\ind{p}+8)\Delta$, $e_q^{\mu}(f)\le 
(5\ind{q}+8)\Delta$, and it is easy to verify that
$e_s^{\mu'}(f')=e_p^{\mu}(f)+e_q^{\mu}(f)\le (5\ind{p}+5\ind{q}+16)\Delta$.
Note that $d_s\ge \max\{d_p,d_q\}-1\ge \frac{d_p+d_q}2-1$, implying
that $e_s^{\mu'}(f')\le \egyd{s}\Delta'$, as required.
\end{proof}

\subsection{Bounding the number of iterations}\label{sec:major}
Recall the notions of shrinking, expanding and neutral
iterations from Section~\ref{sec:iter-weak}.  We shall prove the following bound.
\begin{theorem}\label{thm:major}
The total number of iterations in \textsc{Enhanced Continuous Scaling}
is at most \lengthtotal, among them at most $195n^2m$ shrinking ones.
\end{theorem}

The ground set $V$ changes due to the arc contractions. Let us say
that a node $s$ is {\em born} in iteration $\tau+1$ if $s\in \{p,q\}$ for
an abundant arc contracted in iteration $\tau$; the original nodes are born in
iteration 1. Note that we keep the same notation $p$ or $q$ for the
new node.
Further, we say that a node is {\em alive} until the
first iteration when an
incident arc gets contracted, when it {\em dies}.
 Also note that multiple contractions may happen in the same
iteration; in this case, some nodes die immediately after they are
born; such nodes will be ignored in the analysis.
A key quantity in the analysis is 
\[
\Gamma_i:=\log_2 \frac{32mn\Delta}{|b_i^\mu|},
\]
for all nodes $i\in V-\sink$.
Let $\Gamma_i\itr{\tau}$ denote the value at the beginning of
iteration $\tau$.
We first show that $\Gamma_i\ge 0$ must hold for every $i\in V-\sink$, as otherwise some
abundant arcs would appear.
\begin{claim}\label{cl:whenabundant-mod}
$\Gamma_i\ge 0$ holds for all $i\in V-t$ in every iteration after the first one.
\end{claim}
\begin{proof}
Assume $\Gamma_i\le0$, that is, $|b_i^\mu|\ge32mn\Delta$ holds for some node $i\in V-t$ at a certain iteration after the first one.
We show that there is an abundant incoming or outgoing arc incident to $i$. This contradicts the fact that all such arcs were contracted at the end of the previous iteration.
Since $f$ is generalized flow in every iteration, we have $e_i^\mu\ge 0$. If there are no abundant arcs incident, then  $f^\mu_{ji}<\arcbound$ on every incoming
arc $ji$ and $f^\mu_{ij}<\arcbound$ on all outgoing arcs $ij$. First,
consider the case when $b_i^\mu>0$.
Now
\[
0\le e_i^\mu=\sum_{j:ji\in E}\gamma_{ji}^\mu f_{ji}^\mu-\sum_{j:ij\in
  E} f_{ij}^\mu-b_i^\mu< 17\ind{i}m\Delta-32nm\Delta<0
\]
a contradiction. On the other hand, if $b_i^\mu<0$, then 
\[
(4\ind{i}+8)\Delta\ge e_i^\mu= \sum_{j:ji\in E}\gamma_{ji}^\mu f_{ji}^\mu-\sum_{j:ij\in
  E} f_{ij}^\mu-b_i^\mu>-17\ind{i}m\Delta+32nm\Delta\ge (15nm+17m)\Delta,
\]
using $d_i\le n-1$. This is a contradiction since $m\ge n\ge d_i+1$.
\end{proof}

Let us introduce the following set; recall that 
$\ka$ is the number of  abundant
arcs contracted so far.
\begin{equation}\label{D-def}
D:=\left\{i\in V-\sink: |b_i^\mu|\ge\frac{\Delta}{16^\ka n}\right\}.
\end{equation}
Let $D\itr{\tau}$ denote this set at the beginning of iteration
$\tau$. Note that the condition for calling
\textsc{Filtration} in the algorithm is precisely $(V\setminus T)\cap D=\emptyset$.

\begin{lemma}\label{lem:gamma-change}
\begin{enumerate}[(i)]
\item
The $\Gamma_i\itr{\tau}$ values are monotone decreasing inside every major
cycle, and they increase by 4 when an abundant arc is contracted. 
\item
After the contraction of $\ka$ abundant arcs,
\[
\Gamma_i\itr{\tau}\le 4\ka+5+4\log_2 n
\]
holds for every $i\in D\itr\tau$.
\item
$D\itr{\tau}\subseteq D\itr{\tau+1}$ inside a major cycle. When an abundant
arc $pq$ is contracted at the end of iteration $\tau$, then
$D\itr{\tau}\setminus\{p,q\}\subseteq D\itr{\tau+1}\setminus\{p,q\}$.
\end{enumerate}
\end{lemma}
\begin{proof}
 Inside a major cycle of the algorithm,  the  ratio $|b_i^\mu|/\Delta$
 can never decrease: in \estep$(T,T_0,f,\mu,\Delta)$, it is unchanged for $i\in T$ and
 increases for $i\in V\setminus T$. At the end of a major cycle,
every ratio $|b_i^\mu|/\Delta$
 decreases by a factor of 16. This proves {\em (i)}.
Part {\em (ii)} is straightforward by
$|b_i^\mu|\ge\Delta/(16^\ka n)$ and $\log_2 (mn^2)\le 4\log_2n$.

For part {\em (iii)}, it is straightforward that if no arcs are
contracted, then no node may leave
$D$. Further, when an abundant arc is contracted,
the threshold in the definition of $D$ is unchanged since
$\Delta/16^k=(16\Delta)/16^{k+1}$. Therefore if $i\in
D\setminus\{p,q\}$ before the contraction, then $i$ remains in $D$ after the contraction.
\end{proof}
Let us introduce some further classification of iterations. Let $\cal C$ denote the set of iterations when contractions are
performed. Clearly, $|{\cal C}|\le n-1$.
Let $\cal F$ denote the set of iterations when the subroutine  \textsc{Filtration} is performed; such iterations will be called {\em filtrating}.
 Notice that $\tau\in {\cal F}$, that is, iteration $\tau$ is filtrating if and only if $(V\setminus T\itr{\tau})\cap D\itr{\tau}=\emptyset$.
Let $\cal D$ denote the set of iterations $\tau$ when $D$ is extended:
$D\itr{\tau}\subsetneq D\itr{\tau+1}$. By the above claim, this may
happen at most $2n-1$ times, as every node may enter $D$ only once
during its lifetime. Hence $|{\cal D}|\le 2n-1$. 
Let us define 
\[
\Gamma\itr\tau:=\sum_{i\in D\itr{\tau}} \Gamma_i\itr{\tau}
\]
\begin{claim}\label{cl:gamma-increase}
During the entire algorithm, the total increase in the value of
$\Gamma\itr\tau$ can be bounded by $14n^2$.
\end{claim}
\begin{proof}
When a node $i$ enters $D$ after the contraction of $\ka$ arcs, by
Lemma~\ref{lem:gamma-change}(ii) we have  $\Gamma_i\le 4\ka+5+
4\log_2 n$. There are $\le n-1-\ka$ more contractions, accounting
for a total increase of $\le 4(n-1-\ka)$ in all later
iterations. Hence the total increase for a node $i$ is bounded by
$4n+1+4\log_2 n\le 7n$.
On the other hand, there are altogether $\le 2n-1$ nodes born during
the entire algorithm.
\end{proof}

The following claim is straightforward, since for
 every $i\in V\setminus T$, $b_i^\mu$ is unchanged during \estep$(T,T_0,f,\mu,\Delta)$,
 whereas $\Delta$ decreases by a factor $\alpha$. 

\begin{claim}
If iteration $\tau\notin{\cal F}$, then for at least one $i\in D\itr\tau$, the $\Gamma_i$ value decreases by
$\log_2 \alpha\itr\tau$.
\end{claim}
Together with Claim~\ref{cl:gamma-increase}, it
yields the following.
\begin{lemma}\label{lem:baralpha}
During the entire algorithm, we have
\[
\sum_{\tau\notin{\cal C}\cup{\cal F}} \log_2\alpha\itr\tau\le 14n^2,
\]
\end{lemma}
\begin{proof}
The right hand side bounds the total increase in $\Gamma$
according to Claim~\ref{cl:gamma-increase}. By the previous claim, at
least one $\Gamma_i\itr{\tau}$ 
decreases by at least $\log_2\alpha\itr\tau$ in iteration
$\tau\notin{\cal F}$. 
By Claim~\ref{cl:whenabundant-mod}, $\Gamma_i\itr{\tau+1}\ge 0$ and therefore
$\Gamma_i\itr{\tau}\ge \log_2\alpha\itr\tau$, as otherwise an abundant arc incident to $i$ should have been contracted at the end of iteration $\tau$, giving $\tau\in {\cal C}$.
 \end{proof}

The following lemma is the analogue of Lemma~\ref{lem:sigma-bound}.
\begin{lemma}\label{lem:strong-sigma-bound}
While alive, every node $i\in V-t$ may enter 
the set $T_0$ at most $|{\cal D}|+\sum_{\tau\notin{\cal C}\cup{\cal F}} \log_2\alpha\itr\tau$ times.
\end{lemma}
Before proving the lemma, let us show how it can be used to bound the total
number of iterations.
\begin{proof}[Proof of Theorem~\ref{thm:major}]
The proof follows the same lines as that of Lemma~\ref{lem:iteration-bound}, analyzing the invariant  $\Psi$ as defined by (\ref{def:psi}).
Consider an iteration $\tau\in {\cal C}$ when some abundant arcs are
contracted. According to Claim~\ref{cl:contract-feas}, the value of
$\Psi$ decreases to 0 in all such iterations.

Every shrinking iteration decreases $\Psi$ by one, and the only steps
when $\Psi$ increases is when some node $i\in V-t$ enters $T_0$. Let
$\lambda_i$ denote the number of times this
happens. Lemmas~\ref{lem:baralpha} and \ref{lem:strong-sigma-bound}
imply $\lambda_i\le  |{\cal D}|+14 n^2\le 2n+14n^2\le 15n^2$.
Consequently, the total increase in $\Psi$ is bounded by
\[
\sum_{i\in V-\sink}(3d_i+7)\lambda_i\le 
15n^2\sum_{i\in V-\sink}(3d_i+7)\le 15n^2(6m+7n)\le 195n^2m.
\] 
As in the proof of  Lemma~\ref{lem:iteration-bound}, this bounds the
number of shrinking iterations, and there can be $\le 2n$ iterations
between two subsequent shrinking iterations. This completes the proof.
\end{proof}

The next claims are needed for the proof of Lemma~\ref{lem:strong-sigma-bound}.
\begin{claim}\label{cl:max-e}
Consider a filtrating iteration $\tau\in {\cal F}$. The maximum flow problem in 
 \textsc{Filtration}$(V\setminus T,f,\mu)$ is feasible, and after the subroutine,
every $i\in V\setminus T$ satisfies
\[
e^\mu_i\le R^\mu_i+n\max_{j\in (V\setminus T)-t}|b_j^\mu|.
\]
\end{claim}
\begin{proof}
Feasibility is verified by the restriction of $f\itr{\tau}$ to tight arcs in $E[V\setminus T]$.
This gives a feasible solution as in the proof of
Lemma~\ref{lem:final-opt}; note that the arcs entering $V\setminus T$
are all non-tight, as otherwise we would have extended $T$ in this
iteration instead. Let $f'$ denote the generalized flow on $V\setminus T$ returned by \textsc{Tight Flow}$(V\setminus T,\mu)$, and $f$
the generalized flow returned by \textsc{Filtration}$(V\setminus T,f,\mu)$.
Inside $E[V\setminus T]$, $f$ is nonzero only on tight arcs, and
equals $f_{ij}=f'_{ij}$ for all $ij\in E[V\setminus T]$.
The value of $f$ is set to zero on arcs leaving $V\setminus T$ and the
original values $f\itr{\tau}_{ij}$ are kept if $i\in T$. 
We obtain $e^\mu_i(f)=e^\mu_i(f')+R^\mu_i$ for $i\in V\setminus T$,
since the non-tight arcs are precisely those coming from $T$.
The claim then follows by Theorem~\ref{thm:tight-flow}(ii).
\end{proof}

\begin{lemma}\label{lem:fluvial}
Let $\tau\in {\cal F}\setminus {\cal C}$ be a filtrating iteration when no contraction is performed.
\begin{enumerate}[(i)]
\item If $\beta_i\itr{\tau+1}\ge\negd{i}$ for
  some 
$i\in V\setminus T\itr{\tau}$, then  $\tau\in {\cal D}$, that is, $D\itr{\tau+1}\supsetneq D\itr{\tau}$.
\item Either $\tau\in {\cal D}$, or
one of the
  iterations  $\tau$, $(\tau+1)$ and $(\tau+2)$ must be shrinking.
\end{enumerate}
\end{lemma}
\begin{proof}
{\em (i):} 
Let $\Delta=\Delta\itr{\tau}$ and $T=T\itr{\tau}$.
First, let us prove that \estep$(T,T_0,f,\mu,\Delta)$ must have been performed in iteration $\tau$.
This follows by Claim~\ref{cl:max-e}. Indeed, if \estep$(T,T_0,f,\mu,\Delta)$ is
skipped after calling \textsc{Filtration}$(V\setminus T,f,\mu)$, then for every $i\in V\setminus T$ we have
\[
e_i^\mu\le R_i^\mu+n\max_{j\in (V\setminus T)-t} |b_j^\mu|<(\ind{i}+1)\Delta
\]
at the beginning of iteration $\tau+1$. This follows by
Claim~\ref{cl:R-i-bound} ($R_i^\mu<d_i\Delta$), and since 
$\max_{j\in (V\setminus T)-t}|b_j^\mu|<\Delta/n$ by $(V\setminus
T)\cap D\itr{\tau}=\emptyset$. This is a contradiction to $\beta_i\itr{\tau+1}\ge \negd{i}$.
This shows \estep$(T,T_0,f,\mu,\Delta)$ must have been performed in iteration $\tau$, setting
$\Delta'=\Delta\itr{\tau+1}=\Delta\itr{\tau}/\alpha\itr{\tau}$ (note
that we assumed $\tau\notin {\cal C}$ as well).

Consider a node $i\in V\setminus T$ in iteration $\tau$ for which
$\beta_i$ increased above $\negd{i}$.  
After \textsc{Filtration}$(V\setminus T,f,\mu)$, $F^\mu[V\setminus T]=\emptyset$ and $f_{pq}=0$ for every $pq\in E[V\setminus T,T]$. Therefore \estep$(T,T_0,f,\mu,\Delta)$
does not change the flow $f$ at all; also by definition, the labels $\mu_i$ are
unchanged for $i\in V\setminus T$. 
 Hence 
$e_i^\mu$ and $b_i^\mu$ do not change for $i\in V\setminus T$.
Let $\Delta$ and $\Delta'$ denote
the scaling factor before and after \estep$(T,T_0,f,\mu,\Delta)$. We have
 \[
\ind{i}+1\le \frac{e_i^\mu}{\Delta'}\le
\frac{R^\mu_i+n\max_{j\in (V\setminus T)-t}|b_j^\mu|}{\Delta'}\le
\ind{i}+\frac{n\max_{j\in (V\setminus T)-t}|b_j^{\mu}|}{\Delta'}.
\]
In the second inequality we use that $R_i^\mu$ is unchanged in
\estep$(T,T_0,f,\mu,\Delta)$ and it must be at most $\ind{i}\Delta'$ by Claim~\ref {cl:R-i-bound}.
This implies $\Delta'/n\le \max_{j\in (V\setminus T)-t}|b_j^{\mu}|$.
Since $(V\setminus T)\cap D\itr{\tau}=\emptyset$ was assumed, it
follows that $D$ must be extended in this iteration, that is, $\tau\in
{\cal D}$. 

For part {\em (ii)}, assume $\tau\notin{\cal D}$.
Some nodes $i\in T_0$ might be removed in iteration $\tau$ if $e_i^\mu$
decreases below $\egyd{i}$; in this case, iteration $\tau$ itself is
shrinking. 
Otherwise, part
{\em(i)} implies that $V\setminus T\itr{\tau}\subseteq N\itr{\tau+1}$, and that
$\alpha=\alpha_2$ in iteration $\tau$. Therefore at the
beginning of iteration $\tau+1$, there exists a tight arc $ij\in E$
with $i\in T\itr{t}$, $j\in V\setminus T\itr{t}$. Now either $T\itr\tau\cap
N\itr\tau\neq\emptyset$ already holds, in
which case a path augmentation is performed; or iteration $\tau+1$
extends $T$ using the tight arc $ij$. In this case, $j\in
T\itr{\tau+2}\cap N\itr{\tau+2}$, and  iteration
$\tau+2$ is shrinking.
\end{proof}

We are ready to prove Lemma~\ref{lem:strong-sigma-bound}. The proof is
based on that of Lemma~\ref{lem:sigma-bound}, also making use of the
above claims.

\begin{proof}[Proof of Lemma~\ref{lem:strong-sigma-bound}.]
Let $\tau_1< \tau_2< \ldots <
\tau_\lambda$
denote the iterations when $i$ enters $T_0$. This number is not
necessarily finite; hence $\lambda=\infty$ is allowed.
We have
 $\beta_{i}\itr{\tau_\ell+1}=\ketd{i}$ for
 $1\le \ell\le \lambda$.

For $1\le\ell\le\lambda$, let us define $\tau'_\ell$ to be the largest
value $\tau'_\ell\le \tau_\ell$ such that
$\beta_{i}\itr{\tau'_\ell}< \egyd{i}$. The
existence of these values follows as in the proof of
Lemma~\ref{lem:sigma-bound}. 

In iteration $\tau'_\ell$, we have $i\notin T_0$, since once the
excess $e_i$ drops below 
$\egyd{i}\Delta\mu_i$, the node $i$ is immediately removed from $T_0$.
Also, $i$ will be added to $T_0$ in iteration $\tau_\ell$.
We claim that ${\cal C}\cap
[\tau'_\ell,\tau_\ell]=\emptyset$. Indeed, if $\theta\in {\cal C}$,
then
$\beta_i\itr{\theta+1}<\egyd{i}$ by
Claim~\ref{cl:contract-feas}; this would contradict the maximal choice
of $\tau'_\ell$.

Let us analyze the case when  ${\cal D}\cap
[\tau'_\ell,\tau_\ell]=\emptyset$ holds. According to
Lemma~\ref{lem:fluvial}(i), this implies
${\cal F}\cap
[\tau'_\ell,\tau_\ell]=\emptyset$. Indeed, if $\theta\in {\cal F}\cap
[\tau'_\ell,\tau_\ell]$, then $\beta_i\itr{\theta+1}<\negd{i}$ would
follow, a contradiction again to the maximal choice of $\tau'_\ell$.

With the same argument as in the proof of Lemma~\ref{lem:sigma-bound},
making use of Claim~\ref{cl:e-inc}, we obtain
\[
4\le \left(\prod_{\theta\in [\tau'_\ell,\tau_\ell]} \alpha\itr{\theta}\right)^2.
\]
Note that we have $[\tau'_\ell,\tau_\ell]\cap ({\cal C}\cup {\cal
  F}\cup {\cal D})=\emptyset$.
Let us add the logarithms of these inequalities for those values
$\ell=1,\ldots,\lambda$ where $[\tau'_\ell,\tau_\ell]\cap {\cal
  D}=\emptyset$.
Hence we obtain
\[
\lambda-|{\cal D}|\le \sum_{\theta\notin {\cal C}\cup {\cal F}} \log_2\alpha\itr{\theta}
\]
completing the proof.
\end{proof}

\subsection{Running time analysis}

\begin{proof}[Proof of Theorem~\ref{thm:strong-running}.]
As shown in Theorem~\ref{thm:major}, the total number of shrinking
steps is $O(n^2m)$. If \textsc{Filtration} is not called between two
shrinking iterations, then this part of the algorithm
can be implemented in $O(m+n\log n)$
time using Fibonacci heaps, using the variant described in Section~\ref{sec:rounding}. If \textsc{Filtration} is called, then we
must execute a maximum flow computation in $O(nm)$ time
\cite{Orlin13}. 
According to Lemma~\ref{lem:fluvial}(ii), in this case we must have a
shrinking one within the next three iterations. Consequently,
the running time between two shrinking iterations is dominated by
$O(nm)$.
This gives a total estimation of $O(n^3m^2)$; all other steps of the
algorithm (contractions, initial and final flow computations, etc.)
are dominated by this term.
\end{proof}

\section{Bounding the encoding size}\label{sec:rounding}

In this section, we complete the proof of Theorem~\ref{thm:main}: we modify the algorithm to
guarantee that the encoding size of the numbers during the
computations remain polynomially bounded in the input size.
Further, we present a more efficient implementation, by jointly
performing the elementary steps between two shrinking iterations;
this enables a better running time bound, as already indicated in the
proofs of Theorems~\ref{thm:weak-running} and \ref{thm:strong-running}.
We describe the
modifications for the \textsc{Enhanced Continuous Scaling} algorithm,
but they are naturally applicable for the weakly polynomial
\textsc{Continuous Scaling} algorithm as well.

In this section, let us assume that
\begin{equation}
\bar B\ge
500n^5.\label{nagy-B}
\end{equation}
Indeed, if $\bar B$ is polynomially bounded in $n$, then any of the
previous weakly polynomial algorithms has strongly polynomial running time.
We define the following quantities needed for the roundings; as in the previous
section, $n$ and
$m$ will always refer to the size of the original input instance (and
not the actual contracted one).
\[
q:={40m\bar B^4},\quad \bar q:={40m\bar B^2}=q/\bar B^2
\]
For a real number $a\in \R_{\ge0}$, let
 $\lfloor a\rfloor_q$ denote the largest number $p/q$ with $p\in \Z$,
 $p/q\le a$, and similarly, let
$\lceil a\rceil_q$ denote the smallest number $p/q$ with $p\in \Z$, $p/q\ge a$.
The same notation will also be used for $\bar q$.



\begin{figure}[!htp]
\begin{center}
\vskip-1cm
\fbox{\parbox{\textwidth}{
\begin{tabbing}
xxxxx \= xxx \= xxx \= xxx \= xxxxxxxxxxx \= \kill
\> \textbf{Subroutine} \textsc{Aggregate steps}$(T_0,f,\mu,\Delta)$\\
\> \textbf{for all} $i\in (V\setminus T_0)-\sink$ \textbf{do} \\
\> \> \textbf{update} $r_1(i),r_2(i),r_3(i),r_4(i)$ and $\delta_i$ as
in (\ref{def:f-r})  and (\ref{eq:delta-i}) ; \\
\> \> $\alpha_i\leftarrow
 \min\left\{\lfloor \delta_i\rfloor_q,\min\{1/\gamma_{ji}^\mu: ji\in
   E_f^\mu(\Delta), j\in T_0\}\right\}$ ;\\
\> $\alpha_t\leftarrow\min\{1/\gamma_{jt}^\mu: jt\in E_f^\mu(\Delta), j\in T_0\}$ ;\\
\> \textbf{for all} $i\in T_0$ \textbf{do}\\
\> \> \textbf{update} $\rho_i, \nu_i$ as in (\ref{def:nu-T}) ;
$\alpha_i\leftarrow 1$ ;\\
\> $T\leftarrow T_0$ ; $\alpha^*\leftarrow 1$ ; $\lambda\leftarrow\min\{\nu_j: j\in T_0\}$ ;\\
\> \textbf{while} ($\alpha^*\le \lambda$) and ($t\notin T$) \textbf{do}  \\
\> \> $\alpha^*\leftarrow \min\{\alpha_i: i\in V\setminus T\}$ ;\\ 
\> \> \textbf{if} $\alpha^*=\infty$ \textbf{then} \\
\> \> \> set $f_{ti}=0$ for all $ti\in E$: $\gamma^\mu_{ti}<1$ ;\\
\> \> \> \textbf{return} optimal flow $f$ and optimal relabeling $\mu$
;  \textbf{TERMINATE}.\\
\> \> $i\leftarrow \mbox{argmin}\{\alpha_i: i\in V\setminus T\}$ ;\\ 
\> \> $T\leftarrow T\cup \{i\}$ ;\\
\> \> \textbf{if} $\alpha_i=\lfloor \delta_i\rfloor_q$ \textbf{then}
$T_0\leftarrow T_0\cup \{i\}$ ; \\
\> \> \textbf{for all} $ij\in E: j\in T$ \textbf{do}\\
\> \> \> 
$f_{ij}\leftarrow f_{ij}/\alpha^*$ ;\\
\> \> \> \textbf{update} $\rho_j,\nu_j$ as in (\ref{def:nu-T}) ;\\
\> \> \> $\lambda\leftarrow \min\{\lambda, \nu_j\}$ ; \\
\> \>  \textbf{for all} $ij\in E: j\in V\setminus T$ \textbf{do}\\
\> \> \> \textbf{if} $\gamma_{ij}^\mu<1$ \textbf{then}
$f_{ij}\leftarrow f_{ij}/\alpha^*$ ;\\
\> \> \> \textbf{if} $j\neq t$ \textbf{then}\\
\> \> \> \> \textbf{update} $r_1(j),r_2(j),\delta_j$ as in (\ref{def:f-r}) and
\eqref{eq:delta-i} ;\\
\> \> \> \> $\alpha_j\leftarrow  \min\left\{\alpha_j, \lfloor \delta_j\rfloor_q,\alpha_i/\gamma_{ij}^\mu\right\}$ ;\\
\> \> \> \textbf{if} $j= t$ \textbf{then} $\alpha_t\leftarrow  \min\left\{\alpha_t,\alpha_i/\gamma_{it}^\mu\right\}$ ;\\
\> \>  \textbf{for all} $ji\in E: j\in V\setminus T$ \textbf{do}\\
\> \> \>  \textbf{if} $\gamma_{ji}^\mu=1$ \textbf{then}\\
\> \> \> \>  $f_{ji}\leftarrow f_{ji}\alpha^*$ ; \\
\> \> \> \> \textbf{if} $f_{ji}>\Delta\mu_j$ \textbf{then}
$\alpha_{j}\leftarrow \alpha^*$ ;\\ 
\> \> \> \textbf{if} $j\neq t$ \textbf{then}\\
\> \> \> \> \textbf{update} $r_3(j),r_4(j),\delta_j$ as in (\ref{def:f-r}) and
 \eqref{eq:delta-i} ;\\
\> \> \> \> $\alpha_j\leftarrow  \min\left\{\alpha_j, \lfloor \delta_j\rfloor_q\right\}$ ;\\
\> \>  \textbf{update} $\rho_i,\nu_i$ as in (\ref{def:nu-T}) ;\\
\> \> $\lambda\leftarrow\min\{\lambda,\nu_i\}$ ;\\
\> \>  \textbf{if} $\left(\forall j\in (V\setminus T)-\sink:
  |b^\mu_j|<\frac{\Delta}{16^\ka n(1+1/\bar B)^g}\right)$ \textbf{then}\\
\> \> \> \textsc{Filtration}$(V\setminus T,f,\mu)$ ;\\
\> \> \> \textbf{for all} $j\in T$ \textbf{do} \\
\> \> \> \> $e_j\leftarrow e_j-\rho_j$ ; $\rho_j\leftarrow 0$ ; \\
\> \> \> \> \textbf{update} $\nu_j$ as in (\ref{def:nu-T}) ;\\
\> \> \> \textbf{for all} $j\in V\setminus T$ \textbf{do}\\
\> \> \> \>  \textbf{update} $r_1(i),r_2(i),r_3(i),r_4(i), \delta_i$
as in (\ref{def:f-r}) and \eqref{eq:delta-i} ;\\
\> \> \> \> $\alpha_i\leftarrow \min\left\{\lfloor \delta_i\rfloor_q,\alpha_i\right\}$ ;\\
\> \> \> $\lambda\leftarrow \min\{ \nu_j: j\in T\}$ ; \\
\> $\Delta\leftarrow \lceil \Delta/\alpha^*\rceil_q$ ;\\
\> \textbf{for all} $ij\in F^\mu[V\setminus T]\cup E[V\setminus T,T]$
\textbf{do} $f_{ij}\leftarrow f_{ij}/\alpha^*$ ;\\
\> \textbf{for all} $i\in T$ \textbf{do} $\mu_i\leftarrow
\mu_i\alpha^*/\alpha_i$ ;\\
\> \textbf{for all} $j\in T_0: \nu_j<\alpha^*$ \textbf{do} $T_0\leftarrow
T_0\setminus \{j\}$ ; \\
\> \textsc{Round Label}$(f,\mu)$ ;\\
\> \textbf{if }$t\in T$ \textbf{then RETURN} $s=t$ ;  \textbf{elseif } $\exists s\in T\setminus T_0: \nu_s<\alpha^*$
\textbf{then RETURN} such an $s$ ;\\
\> \textbf{otherwise RETURN }$s=NULL$.
\end{tabbing}
}}
\caption{The Aggregate steps subroutine}\label{code:block}
\end{center}
\end{figure}

\begin{figure}[htb]
\begin{center}
\fbox{\parbox{\textwidth}{
\begin{tabbing}
xxxxx \= xxx \= xxx \= xxx\= xxx \= xxx \= xxxxxxxxxxx \= \kill
\> \textbf{Algorithm} \textsc{Modified Enhanced Continuous Scaling}\\
\> \textsc{Initialize}; \\
\> $T_0\leftarrow \emptyset$ ; $\ka\leftarrow 0$ ;  $g\leftarrow 0$ ; \\
\> \textsc{While} $|V|>1$ {\em and} $\Delta\ge \threshold$ \textsc{do}\\
\> \> $s\leftarrow$ \textsc{Aggregate Steps}$(T_0,f,\mu,\Delta)$ ; \\
\> \> $g\leftarrow g+1$ ; \\
\> \> \textbf{if} $s\neq NULL$ \textbf{then} \\
\> \> \>  \textbf{pick} a  tight $p-s$ path $P$ in $E^\mu_f(\Delta)$
with $p\in T_0$ ; \\
\> \> \> \textbf{send} $\Delta$ units of relabeled flow from $p$ to $s$ along $P$;\\
\> \> \> \textbf{if} $e_p^\mu<\egyd{p}\Delta$ \textbf{then} $T_0\leftarrow
T_0\setminus \{p\}$ ;\\
\> \> \textbf{while} $\exists\ pq\in E$: $f_{pq}^\mu\ge \arcbound$
\textbf{do} \\
\> \> \> \textbf{for all} $ij\in E:\ \gamma_{ij}^\mu<1$ \textbf{do}
$f_{ij}\leftarrow 0$ ;\\
\> \> \> \textsc{Contract}($pq$) ;\\
\> \> \> $\Delta\leftarrow 16\Delta$ ;\\
\> \> \> $\ka\leftarrow \ka+1$ ;\\
\> \> \> $T_0\leftarrow \emptyset$ ; \\
\> \textbf{if} $\Delta<\threshold$  \textbf{then} \textsc{Tight-Flow}$(V,\mu)$ \\
\> \> \textbf{else}  \textsc{Expand-to-Original}$(\mu)$ ;
\end{tabbing}
}}
\caption{Description of the modified strongly polynomial
  algorithm}\label{code:modified strong}
\end{center}
\end{figure}

The  main subroutine \textsc{Aggregate steps}$(T_0,f,\mu,\Delta)$ is shown in Figure~\ref{code:block}. 
The input is a set $T_0$ with $e_i^\mu\ge\egyd{i}\Delta$ for every
$i\in T_0$. The output is either a node $s\in V$ or $s=NULL$. If
$s\neq NULL$, then we can find a tight path in $E_f^\mu(\Delta)$ between a node $p\in T_0$
and  $s$, where either $s=t$ or $e_s^\mu<\negd{s}\Delta$.
The case $s=NULL$ means that some node $i$ leaves $T_0$, that is, $e_i^\mu$ drops below
$\egyd{i}\Delta$.
In the first case, we can perform a path augmentation from $s$ to $t$. The entire
algorithm is exhibited in Figure~\ref{code:modified strong}. Note that
termination can happen either because the graph is shrunk to a single
node, or because $\Delta$ decreases below a certain threshold as in the weakly
polynomial algorithm
\textsc{Continuous Scaling}.

We now describe the main features of the subroutine \textsc{Aggregate
  steps} and compare it to the algorithm \textsc{Enhanced Continuous
  Scaling}.
Apart from the rounding and contraction steps, it performs essentially the same as a
sequence of \estep s  starting with $T=T_0$, until a next
shrinking iteration. The difference is that in \textsc{Enhanced Continuous
  Scaling}, whenever the set $T$ is extended by a node, \estep{} needs to update the
labels of every node in $T$ and change flow values on certain
arcs. During a sequence of iterations between two shrinking steps,
this can lead to $O(n)$ value updates in certain nodes and arcs. 
In contrast, \textsc{Aggregate steps} changes the labels only once and
the flow
values at most twice. The key quantity in \textsc{Aggregate steps}  is
$\alpha^*$. This corresponds to the
product of the $\alpha$ multipliers of the sequence of \estep s thus
far in \textsc{Enhanced Continuous Scaling}. 

The subroutine \textsc{Aggregate steps} starts with $T=T_0$ and extends $T$ by adding
nodes one-by-one, until  $t\in T$, or $e_s^\mu<\negd{s}\Delta$
for some $s\in T\setminus T_0$, or $e_i^\mu<\egyd{i}\Delta$ for some $i\in T_0$.
 Node labels are changed at the end of the
subroutine only. 
For a node $i\in T$, we store the value
$\alpha_i$ at the time when $i$ enters $T$, and multiply the label
$\mu_i$ at the end by $\alpha^*/\alpha_i$. The flow on an arc $ij$ may
change when its endpoints enter $T$, or at the end of the subroutine,
altogether at most twice.

For nodes $i\in  V\setminus T$, we use $\alpha_i$ to denote the candidate value of
$\alpha^*$ when $i$ must enter $T$, either due to a new tight arc
$ji\in E_f^\mu(\Delta)$, $j\in T$, or because the excess $e_i$ reaches
the threshold
$\ketd{i}\Delta\mu_i$ and hence $i$ must be included in $T_0$. We 
define $\delta_i$ as in \eqref{eq:delta-i},  representing the value of $\alpha^*$
when $i$ would enter $T$ because of $e_i=\ketd{i}\Delta\mu_i$,
provided that no other node enters $T$ before. We let
\[\alpha_i:=
 \min\left\{\lfloor \delta_i\rfloor_q,\min\{\alpha_j/\gamma_{ji}^\mu: ij\in
   E_f^\mu(\Delta), j\in T\}\right\}.
\]
Note the rounding $\lfloor \delta_i\rfloor_q$ in the first case. This
means that $e_i$ might be slightly less than $\ketd{i}\Delta\mu_i$
when $i$ enters $T$.
The second event corresponds to the case when $i$ enters $T$ due to a
new tight arc from a node $j\in T$. Note that either $ji\in E$, or $ji$ is a reverse arc
with $ij\in E$, $\gamma_{ij}^\mu=1$,  $f_{ij}^\mu> \Delta$. In the latter
case the corresponding term equals $\alpha_j$.

For $i\in T\setminus T_0$, wish to estimate
when $e_i<\negd{i}\Delta\mu_i$ would be attained, and for $i\in T_0$,
we wish to estimate when $e_i<\egyd{i}\Delta\mu_i$ would be attained.
To provide a unified notation for these two cases, let us define $\xi_i=1$ if $i\in
T\setminus T_0$ and
$\xi_i=2$ if $i\in T_0$. We let
\begin{equation}\label{def:nu-T}
 \begin{aligned}
& \rho_i:= \sum_{j\in V\setminus T} \gamma_{ji}f_{ji},\\
& \nu_i :=\begin{cases}
\infty &\quad \mbox{ if }e_i-\rho_i\ge (d_i+\xi_i)\Delta\mu_i ;\\
 \frac{\rho_i}{(d_i+\xi_i)\Delta\mu_i+\rho_i-e_i} &\quad \mbox{ otherwise.}
\end{cases} 
\end{aligned}
\end{equation}
Here $\rho_i$ denotes the total flow entering $i$ on arcs from
$V\setminus T$, these are the ones where the flow value will be
reduced.
We define 
$\nu_i$ as the smallest value of $\alpha^*$ when
$e_i=(d_i+\xi_i)\Delta\mu_i$ is reached. $\lambda$ will denote the minimum
value of $\{\nu_i: i\in T\}$. The iterations terminate once
$\lambda<\alpha^*$.

In every iteration, we set the new value $\alpha^*:=\min\{\alpha_i:
i\in V\setminus T\}$, pick a node $i$ minimizing this value, and
include it into $T$.
The modifications of the incident $f_{ij}$ and $f_{ji}$ values are in order to guarantee the
same change as in the  sequence of
\estep{} operations in \textsc{Enhanced Continous Scaling}. 
We update the corresponding $r_1(j),\ldots,r_4(j),\delta_j$ and $e_j,\rho_j,\nu_j$ values on the
neighbours of $i$ accordingly. These updates can be performed in
$O(1)$ time. Indeed, for each of the sums
$r_1(j),\ldots,r_4(j),e_j,\rho_j$, only one term changes. Provided
these, $\delta_j$ and $\nu_j$ are obtained by simple formulae.

The \textsc{Filtration} subroutine is used similarly as in
\textsc{Enhanced Continuous Scaling}, but the bound on $|b_j^\mu|$ is
different, containing a term $(1+1/\bar B)^g$ necessary due to the roundings.
The counter $g$ denotes the total number of times the subroutine
\textsc{Aggregate steps} was performed.

Compared to \textsc{Enhanced Continuous Scaling}, there is a further minor difference
regarding contractions. In \textsc{Enhanced Continuous Scaling},  a contraction
can be performed after every \estep, whereas in the modified algorithm,
only after the entire sequence represented by  \textsc{Aggregate steps}$(T_0,f,\mu,\Delta)$.

 If \textsc{Filtration} is not called, then the subroutine
 \textsc{Aggregate Steps} can be
 implemented in $O(m+n\log n)$ time using the Fibonacci heap data structure. To see this, we
 maintain two heap structures, one for the $\alpha_i$'s for $i\in V\setminus T$, an one for the
$\nu_i$'s, $i\in T$. Besides, we maintain the $r_1(i),\ldots,r_4(i)$
values for $i\in V\setminus T$, and the $e_i,\rho_i$ values for $i\in
T$. Every arc is examined $O(1)$ times, and the corresponding key
modifications can be implemented in $O(1)$ time. Consequently, the
bound in \cite{Fredman87} is applicable.

It is easy to verify that all $\mu_i$ and $f_{ij}$ values are modified exactly as in
a sequence of \estep\ operations. For example, consider an arc $ji$
with originally $i,j\in V\setminus T$, such that $i$ enters $T$ before
$j$. The scaling factor when $i$ enters $T$ is $\Delta/\alpha_i$. 
If $ij\in E_f^\mu(\Delta/\alpha_i)$, that is, $f_{ji}^\mu>
\Delta/\alpha_i$, then $j$ enters $T$ in the next neutral phase.
Accordingly, \textsc{Aggregate Steps} sets $\alpha_j=\alpha^*$ in the same case.
If $ji$ was a non-tight arc already at the beginning, then $f_{ji}$ is
decreased in every elementary step until $j$ enters $T$; in our subroutine, 
$f_{ji}$ is divided by $\alpha_j$. However, if $ji$ was tight initially, and 
$f_{ji}^\mu<\Delta\alpha_i$, then it becomes non-tight after $i$
enters $T$. Notice that in this case our subroutine divides
$f_{ji}$ by $\alpha_j/\alpha_i$. The other cases can be verified similarly. 


\medskip

At termination, we perform the subroutine \textsc{Round Label},
shown in Figure~\ref{code:round}. This is a Dijkstra-type algorithm that takes  labeling $\mu$, and changes it to a labeling $\mu'\ge \mu$ such that the set of tight arcs in $E_f$ may only increase. Consequently, if $(f,\mu)$ is a $\Delta$-feasible pair for some $\Delta$, then so is $(f,\mu')$. 

We repeatedly extend the set $S$ starting from $S=\{t\}$ until $S=V$ is achieved.
In every iteration we multiply all $\mu_i$'s for $i\in V\setminus S$ by $\varepsilon>1$, so that either a new tight arc between $V\setminus S$ and $S$ is created, or some value $\mu_i$ for $i\in V\setminus S$ becomes an integer multiple of $1/\bar q$.

\begin{figure}[!htp]
\begin{center}
\fbox{\parbox{\textwidth}{
\begin{tabbing}
xxxxx \= xxx \= xxx \= xxx \= xxxxxxxxxxx \= \kill
\> \textbf{Subroutine} \textsc{Round Label}$(f,\mu)$\\
\> $S\leftarrow \{\sink\}$ ;\\
\> \textbf{while} $S\neq V$ \textbf{do}\\
\> \> $\varepsilon_1\leftarrow \min\left\{\frac{\lceil \mu_i\rceil_{\bar q}}{\mu_i}: i\in V\setminus S\right\}$ ;\\
\> \> $\varepsilon_2\leftarrow\min\left\{\frac{1}{\gamma_{ij}^\mu}:
  ij\in E_f, i\in V\setminus S, j\in S\right\}$ ;\\
\> \> $\varepsilon\leftarrow\min\{\varepsilon_1,\varepsilon_2\}$ ;\\
\> \> \textbf{for} $i\in V\setminus S$ \textbf{do} $\mu\leftarrow \mu/\varepsilon$ ;\\
\> \> $S\leftarrow S\cup \{i\in V\setminus S: \lceil \mu_i\rceil_{\bar q}=\mu_i\}\cup
\{i\in V\setminus S: \exists j\in S, ij\in E_f, \gamma_{ij}^\mu=1\}$ ;
\end{tabbing}
}}
\caption{The Round Label subroutine}\label{code:round}
\end{center}
\end{figure}

\subsection{Analysis}
 It is easy to adapt 
Theorem~\ref{thm:tight-flow} and  Lemma~\ref{lem:final-opt} to show that if
in any contracted graph during the algorithm \textsc{Modified Enhanced Continuous Scaling},
we have $\Delta\le \threshold$ for the original values of $B$, $m$ and $n$, then the current labeling $\mu$ is optimal and thus we may terminate.
Also note that $2\bar B/q\le \threshold$, and therefore we may assume that $2\bar B/q\le \Delta$ in all iterations of the algorithm except the last one.

\begin{claim}\label{cl:round-label}
The subroutine \textsc{Round Label} returns a labeling $\mu'$ such
that every $\mu'_i$ is an integer multiple of $\bar B/q$. If $(f,\mu)$ is $\Delta$-conservative for some $\Delta\ge 0$, then so is $(f,\mu')$. 
Finally, $\mu_i\le \mu'_i\le \left(1+1/({40m\bar B})\right)\mu_i$.
\end{claim}
\begin{proof}
A node $i$ enters $S$ either if $\mu_i$ is an integer multiple of
$1/\bar q=\bar B^2/q$, or if it is connected by a tight path $P$ in
$E_f$ to a node $j$ such that $\mu_j$ is an integer multiple of
$1/\bar q$. In the latter case, $\mu_i=\mu_j/\gamma(P)$, and since
$\bar B$ is an integer multiple of $\gamma(P)$ by definition of $\bar
B$ in Section~\ref{sec:encoding}, it follows that $\mu_i$ is an integer multiple of $\bar B/q$. The claim on conservativeness follows since every tight arc in $E_f$ remains tight.
Finally, it is clear that $\mu'_i\le \lceil \mu_i\rceil_{\bar q}
<\mu_i+1/\bar q=\mu_i(1+1/(\bar q\mu_i))$. On the other hand,
$\mu_i\ge 1/\bar B$ because of the initial definition \eqref{def:mu}, and hence $1+1/(\bar q\mu_i)\le \left(1+1/({40m\bar B})\right)$.
\end{proof}

In the original algorithm, $\Delta\mu_i$ is nonincreasing during every
\estep{} iteration. Due to the roundings, this is not true anymore;
however, we have the following bound (the possible increase
corresponds to the case $\alpha_i\le 1+1/{\bar B}$).
\begin{claim}
When performing \textsc{Aggregate steps}, $\Delta\mu_i$ decreases by
at least a factor of $\alpha_i/(1+1/{\bar B})$ for every $i\in T$, and
by $\alpha^*/(1+1/{\bar B})$ for every $i\in V\setminus T$, except for
possibly the ultimate iteration.
\end{claim}
\begin{proof}
Without the rounding, we would set the new value of the scaling factor
to $\Delta/\alpha^*$ and the new value of $\mu_i$ as
$\mu_i\alpha^*/\alpha_i$ if $i\in T$ and leave it unchanged if $i\in
V\setminus T$. Let us focus on the case $i\in T$; the same 
argument works for $i\in V\setminus T$ as well.
These will be rounded to $\Delta'=\lceil
\Delta/\alpha^*\rceil_q$ and $\mu'_i\le (1+1/({40m\bar
  B}))\mu_i\alpha^*/\alpha_i$ by the previous claim.
As remarked above, we have $2\bar B/q\le \Delta'$ in all save the last
step of the algorithm. Therefore 
$\Delta'=\lceil \Delta/\alpha^*\rceil_q\le (1+{1}/({2\bar
  B}))\Delta/\alpha^*$. Consequently,
\[
\Delta'\mu_i'\le \left(1+ {1}/({2\bar
  B})\right)\left(1+1/({40m\bar B})\right)\Delta\mu_i/\alpha_i\le \left(1+1/{\bar B}\right)\Delta\mu_i/\alpha_i,
\]
proving the claim.
\end{proof}

Provided this, one can derive the bound $O(n^2m)$ on the total number
of calls to \textsc{Aggregate steps} as in Theorem~\ref{thm:major}.
This subroutine corresponds to a sequence of \estep{}, however, 
the argument can be easily adapted.
We now outline the changes in the analysis. Instead of (\ref{D-def}),
we define the set $D$ as
\[
D:=\left\{i\in V-\sink:
  |b_i^\mu|\ge\frac{\Delta}{16^\ka n(1+{1}/{\bar B})^g }\right\}.
\]
According to the above claim, if no arc is contracted, then no node
may leave the set $D$, as in Lemma~\ref{lem:gamma-change}.
After the contraction of $\ka$ arcs, the maximum value of $\Gamma_i$
can be at most
\[
\Gamma_i\itr{\tau}\le 4\ka+5+4\log_2 n+g\log\left(1+{1}/{\bar
    B}\right)\le 4\ka+5+4\log_2 n+{g}/{\bar B}.
\]
By the assumption (\ref{nagy-B}), the last term is at most $1/n$ even
after $500n^2m$ iterations. Hence the proof of
Claim~\ref{cl:gamma-increase} can be easily modified to prove the following.
\begin{claim}\label{cl:gamma-increase-2}
After at most  $500n^2m$ executions of \textsc{Aggregate steps},
the total increase in the value of
$\Gamma\itr\tau$ can be bounded by $14n^2$.
\end{claim}
Another change in the argument is due to the fact that when a node $i$
enters $T_0$ in \textsc{Aggregate steps}, it might have
$e_i<\ketd{i}\mu_i$ due to the rounding of $\delta_i$. This affects
the way  Claim~\ref{cl:e-inc} is applied in the proof of
Lemmas~\ref{lem:sigma-bound} and \ref{lem:strong-sigma-bound}. 
In \eqref{eq:four}, $4$ has to be replaced by a slightly smaller
number; consequently, we have to replace $\log_2$ by
$\log_{2-\varepsilon}$ in the argument for some small
$\varepsilon$. However, this increases the running time estimation
only by a small constant factor.

\medskip

One can show that the $O(n^3m^2)$ bound on the number of
elementary arithmetic operations and comparisions is still applicable for the
modified algorithm. The proof of Theorem~\ref{thm:main} is complete by
showing that the size of the variables remain polynomially bounded.
Due to the rounding steps, $\Delta$ and the $\mu_i$'s are always of
polynomially bounded size. It is left to show that the same holds for
the $f_{ij}$ values. 

\begin{lemma}
Every $f_{ij}$ value is a rational number of polynomially bounded size in $\bar B$.
\end{lemma}
\begin{proof}
The $f_{ij}$ values can be changed in two ways. One is via maximum
flow computations in the initial \textsc{Tight-flow} subroutine and
during the later \textsc{Filtration} iterations. We can always assume
that the flow computations return a basic optimal solution; since the
flow problem is defined by polynomially bounded capacities and
demands, such steps reset a polynomially bounded rational value for
$f_{ij}$.

Every \textsc{Aggregate Steps} iteration
 either leaves $f_{ij}$
unchanged, or modifies it to $f_{ij}/\alpha_i$, or to
$f_{ij}\alpha_j/\alpha_i$. We claim that $\alpha_i$ and $\alpha_j$ are both
integer multiples of $1/q$.
Indeed, either $\alpha_i=\lfloor \delta_i\rfloor_q$ and thus this 
 property is straightforward; or $\alpha_i=\mu_p/\gamma(P)$ for some
 $p-i$ path $P$ with $p\in T_0$; note that $\mu_i$ is an integer
 multiple of $\bar B/q$ by Claim~\ref{cl:round-label}, and $\bar B$ is
 an integer multiple of $\gamma(P)$.
Further, it is easy to verify that $\alpha_i,\alpha_j\le \bar B^2$.
Consequently, $f_{ij}$ is multiplied in  \textsc{Aggregate
  steps}$(T_0,f,\mu,\Delta)$ by a number $Q$ that is the quotient of two integers
$\le q\bar B^2$.

During a path augmentation, $f_{ij}$ is modified by adding or
subtracting $\Delta\mu_i$, that is an integer multiple of $\bar
B/q^2$. Since \textsc{Aggregate Steps} is executed $O(n^2m)$ times,
these arguments show that all $f_{ij}$'s remain polynomially bounded.
\end{proof}

\section{Problem transformations}\label{sec:further}
\subsection{Transformation to an uncapacitated
  instance}\label{sec:transform}
Consider an instance $(V',E',t',u',\gamma')$ of the standard formulation (\ref{primal-standard}) with $|V'|=n'$, $|E'|=m'$, and
encoding parameter $B$.  We now show how it can be transformed to an equivalent instance
$(V,E,t,b,\gamma)$ of the uncapacitated formulation
(\ref{primal}) with $|V|\le n'+m'$, $|E|\le 2m'$, and $\bar B\le  2{B}^{4m'}$ satisfying
 assumptions (\ref{cond:root}), (\ref{cond:init}),
 (\ref{cond:bounded}), and all
 assumptions on the encoding size in Section~\ref{sec:encoding}. 
The transformation proceeds in three steps. First, we remove all arc
capacities by introducing new nodes for arcs with finite
capacities. In the second step, the 
 boundedness condition  (\ref{cond:bounded}) is checked; if the problem turns out to be unbounded,
 we terminate by returning the optimum value $\infty$. Finally,
 new auxiliary arcs are added in order to satisfy (\ref{cond:root}).

\subsubsection*{Removing arc capacities}
Let us divide the arc set as $E'=E'_u\cup E'_\infty$, where $e\in
E'_u$ if the capacity $u'_e$ is finite, and $e\in E'_\infty$ if $u'_e=\infty$.
Let the node set $V$ consist of the original node set $V'$ and a new
node corresponding to every arc $e\in E'_u$; let $t:=t'$.
The original nodes are called {\em primary} nodes, and those corresponding to arcs {\em secondary} nodes.
Let $k=a_{ij}$ be the node corresponding to  arc $ij\in E'_u$. The transformed
graph contains two corresponding arcs, $ik$ and $jk$. We leave all arcs in
$ij\in E'_\infty$ unchanged between the primary nodes $i$ and $j$.
Let us define $\bar B$ to be twice
the product of the numerators and
denominators of all rational numbers $\gamma'_{ij}$ for every $ij\in
E'$ and $u'_{ij}$ for every
$ij\in E'_u$; clearly, $\bar B\le 2B^{4m'}$.

For a primary node $i\in V$, let us set the node demand
$b_i= -\sum_{j:ji\in E'_u} \gamma'_{ji}u'_{ji}$.
For the secondary node $k=a_{ij}$, let $b_k:=\gamma'_{ij}u'_{ij}$. 
Furthermore, let us define the gain factors by 
$\gamma_{ik}:=\gamma'_{ij}$, $\gamma_{jk}:=1$. For $ij\in E'_\infty$,
we let $\gamma_{ij}:=\gamma'_{ij}$.

The transformed instance satisfies (\ref{cond:init}), since the
following $\bar f$ is a feasible solution. For every secondary node $k=a_{ij}$, let us
set
$\bar f_{jk}:=\gamma'_{ij}u'_{ij}$, and let us set $\bar f_{pq}=0$ for
all other arcs $pq$.

\subsubsection*{Boundedness}
Let us now address the boundedness of the problem. The following lemma gives a simple
characterization of boundedness of the objective.

\begin{lemma}\label{lem:boundedness}
Consider a problem instance  $(V,E,t,b,\gamma)$ in the uncapacitated
formulation (\ref{primal}) that is feasible.
The objective in  (\ref{primal}) is bounded if and only if there is no
cycle $C\subseteq E$ with $\gamma(C)>1$ and a path $P\subseteq E$
between a node incident to $C$ and $\sink$.
\end{lemma}
\begin{proof}
If such a cycle exists, then we can increase the flow value in $t$
arbitrarily by generating flow on $C$ and sending it to $t$ via $P$.
For the converse direction, consider the dual program (\ref{dual});
recall that the labels $\mu_i$ are simply the inverses of the dual
variables. 
Since (\ref{primal}) is feasible according to  (\ref{cond:init}), the
objective is bounded if and only if (\ref{dual}) is feasible.
Assume $C\subseteq E$ is a cycle such that a path $P\subseteq E$
connects a node incident to $C$ to $t$.
Using the condition $\gamma_{pq}\mu_p/\mu_q\le 1$ on every arc $pq\in P$, it follows that in every
feasible labeling, $\mu_i$ is finite for every node $i$ incident to $C$.
 Therefore $\gamma(C)=\gamma^\mu(C)\le 1$, completing the proof.
\end{proof}

Let $V'$ denote the set of nodes $i$ such that there exists an $i-t$
path in $E$. This set $V'$ can be found by a simple search algorithm.
Boundedness can be decided by checking for a flow generating
cycle in the restriction of $G$ to $V'$. This is equivalent to finding a negative
cycle for the cost function $c_{ij}=-\log \gamma_{ij}$ and can be solved
by any
negative cycle detection algorithm, see e.g. \cite[Chapter 5.5]{amo}. 
Computations with logarithms can be avoided by devising a multiplicative analogue
of these algorithms working directly with the $\gamma_{ij}$'s.

\medskip

After removing the arc capacities, we run this algorithm to decide boundedness. If the problem is
unbounded, we terminate with optimum value $\infty$. Otherwise, we can
assume the validity of  (\ref{cond:bounded}). (Note that since all
secondary nodes have only two incoming arcs incident, all arcs used in
$C$ and $P$ are necessarily from $E_\infty$. Therefore, the same
subroutine could also be performed before the transformation.)

\subsubsection*{Auxiliary arcs}

To satisfy (\ref{cond:root}), for every node $i\in V-\sink$ for which
$it\notin E$, 
let us further add an arc $it$ to $E$ with $\gamma_{it}:=1/\bar B$.  Let
us call these  {\em auxiliary arcs}. 

The following lemma justifies our transformation.
\begin{lemma}\label{lem:transform-prop}
The transformed instance satisfies assumptions
(\ref{cond:root}), (\ref{cond:init}) and (\ref{cond:bounded}), and $\bar B$ satisfies the  assumptions on the
 encoding sizes in  Section~\ref{sec:encoding}.
 An optimal solution $f$ to the modified problem can be transformed
to an optimal solution $f'$ to the original problem in $O(m')$ time.
\end{lemma}
\begin{proof}
The first part is straightforward. For the second statement, let $f$ be an optimal solution to the modified problem with
an optimal labeling $\mu$ as in Theorem~\ref{thm:genflow-opt}(i).
For a secondary node $k=a_{ij}$, 
let us set $f'_{ij}:=f_{ik}$. 
Let $S_0\subseteq V$ denote the set of nodes $i\in V$ 
for which $\gamma_{it}^\mu=1$, that is $\mu_i=\bar B$.
Let $S\subseteq V$ denote the
set of nodes that can be reached from $S_0$ on a residual path
$P\subseteq E_f$.

Let $S'\subseteq V'$ denote the set of primary nodes in $S$.
Let us set $\mu'_i:=\mu_i$ if $i\in V'\setminus S'$ and 
$\mu'_i:=\infty$ if $i\in S'$.
In what follows, we shall verify the optimality conditions in  Theorem~\ref{thm:genflow-opt}(ii) for $f'$ and $\mu'$.

We first claim that $f'_{ij}\le u'_{ij}$ for all arcs $ij\in
E'_u$. This follows since for the secondary
node $k=a_{ij}$ we have $b_k=\gamma'_{ij}u'_{ij}$, and  $e_k(f)=0$ due
to the optimality of $f$.
Next, we claim that $\sink\notin S$ and therefore $\mu'_\sink=1$. Indeed, assume for a
contradiction there exists a path $P\subseteq E_f$ from a node $i\in
S_0$ to \sink.  Then $\mu_i\le 1/\gamma(P)< \bar B$ by the definition
of $\bar B$, a contradiction to $\mu_i=\bar B$. 

The condition on arcs $ij\in E'[S']$ is straightforward since
$\mu'_i=\mu'_j=\infty$. Consider an arc $ij\in E'$ with $i\in S'$,
$j\in V'\setminus S'$. If $ij\in E'_\infty$, then $ij\in E_f$,
contradicting the definition of $S$. Hence $ij\in E'_u$; 
 let $k=a_{ij}$ be the corresponding
secondary node. By definition, $ik\in E\subseteq E_f$.  By the
definition of $S'$, we must have $kj\notin E_f$, that is,
$f_{jk}=0$ and therefore $f'_{ij}=u'_{ij}$ due to the constraint $e_k(f)= b_k$.  Then
$\gamma_{ij}\mu_i=\infty>\mu_j$, as required. It follows similarly that $f_{ij}=0$
for all arcs $ij\in E'$ with $i\in V'\setminus S'$,
$j\in  S'$, and they satisfy $\gamma_{ij}\mu_i<\infty=\mu_j$. 

Let us focus on arcs $ij\in E'[V'\setminus S']$; assume $ij\in E'_u$ and
$0<f'_{ij}<u'_{ij}$. This means that for the corresponding secondary
node $k=a_{ij}$, we had $f_{ik},f_{jk}>0$, and thus
$\gamma_{ij}\mu_i=\mu_k$, and $\mu_k=\mu_j$, implying
$\gamma_{ij}\mu'_i=\mu'_j$. 
Note that $k\notin S_0$ and $e_k(f)=0$ implies that $f_{ij}\le u'_{ij}$, therefore 
$f'_{ij}=f_{ij}$ on all such arcs. The other cases, including the case
of arcs in $E'_\infty$, follow similarly. 

It is left to prove that $e_i(f')=0$ whenever $i\in V'\setminus S'$.
By definition, $i\notin S_0$ and hence $f_{it}=0$. 
For every  incoming arc $ji$ 
with secondary node $k=a_{ji}$, we have
$f_{jk}=\gamma'_{ji}(u'_{ji}-f'_{ji})$.
Together with $e_i(f)=0$ and
the definition of $b_i$, this implies $e_i(f')=0$. 
\end{proof}

\subsection{Linear programs with two nonzeros per column}\label{sec:lp2}
In this Section, we show how our algorithm can be used to solve
arbitrary linear feasibility problems of the form \eqref{lp2}.

The main part of this argument was given by Hochbaum
\cite{hochbaum04}, showing how an arbitrary instance of \eqref{lp2}
can be transformed to another one where every column of the matrix
$A$ contains exactly one
positive entry and a $-1$ entry. 
With the rows corresponding
to nodes and the columns to arcs, let us use $\gamma_{ij}>0$ to denote
the
positive entry in row $j$ and column $ij$.
 (The construction creates two copies
of the vertex set, and columns with two positive or two negative
entries are represented by two arcs crossing between the copies,
whereas columns with two different signs are represented by two arcs,
one in each copy.)
 The transformed version
may contain upper capacities on the arcs. These can be removed using
the same construction as in Section~\ref{sec:transform}, at the cost
of increasing the number of nodes to $O(m)$.
After removing the arc capacities, we can
write the system in the form
\begin{align}
\sum_{j:ji\in E}\gamma_{ji}f_{ji}-\sum_{j:ij\in E}f_{ij}&= b_i\quad
\forall i\in V\tag{$LP2M$}\label{lp2m}\\
f\ge 0&\notag.
\end{align}
Given an instance of  \eqref{lp2m}, let the value $\bar B$ be chosen as an integer multiple
of the products of all numerators and denominators of the
$\gamma_{ij}$ values, and furthermore, assume $|b_i|\le \bar B$ and
$b_i$ is an integer multiple of $1/\bar B$ for all $i\in V$.

\eqref{lp2m} is an uncapacitated generalized
flow feasibility problem, where all node demands must be exactly
met ($M$ stands for monotone, following Hochbaum's terminology.)  Compared to the formulation \eqref{primal}, the differences are as
follows: {\em (i)} \eqref{lp2m} is a feasibility problem and does not have
a distinguished sink node, in contrast to the
optimization problem \eqref{primal}; {\em (ii)} the node demands must be
exactly met in \eqref{lp2m}, whereas in
\eqref{primal}, nodes are allowed to have excess.
For this reason, we introduce two relaxations of \eqref{lp2m} with
inequalities.
\begin{align}
\sum_{j:ji\in E}\gamma_{ji}f_{ji}-\sum_{j:ij\in E}f_{ij}&\ge b_i\quad
\forall i\in V\tag{$LP2M_\ge$}\label{lp2mg}\\
f\ge 0&\notag\\
~\notag\\
\sum_{j:ji\in E}\gamma_{ji}f_{ji}-\sum_{j:ij\in E}f_{ij}&\le b_i\quad
\forall i\in V\tag{$LP2M_\le$}\label{lp2ml}\\
f\ge 0&\notag
\end{align}
Our main insight (Lemma~\ref{lem:both-feas} below) is that if both
these relaxations are feasible, then \eqref{lp2m} is also feasible,
and a solution can be found efficiently provided the solutions to the
relaxed instances.

The second relaxation  \eqref{lp2ml}
can be reduced to \eqref{lp2mg} by reversing all arcs in $E$, setting
$\gamma_{ji}=1/\gamma_{ij}$ on the reverse arc $ji$ of $ij\in E$, and
changing the node demands to $-b_i$. We show that \eqref{lp2mg} -- and
consequently, \eqref{lp2ml} -- can be solved using our algorithm for \eqref{primal}.

\subsubsection*{Solving  \eqref{lp2mg}}
As a preprocessing step, we identify the set $Z$ of nodes
that can be reached via a path in $E$ from a flow generating cycle in
$E$. That is, $i\in Z$ if there exists a cycle $C\subseteq E$,
$\gamma(C)>1$, and a path $P\subseteq E$ connecting a node of $C$ to
$i$. This set $Z$ can be found efficiently using algorithms for
negative cycle detection, similarly as in
Section~\ref{sec:transform}. Using the flow generating cycles, arbitrary
demands $b_i$ for $i\in Z$ can be met. This solves \eqref{lp2mg} if
$Z=V$; in the sequel let us assume $V\setminus Z\neq\emptyset$.
By the definition of $Z$, there
is no arc in $E$ between $Z$ and $V\setminus Z$. If \eqref{lp2mg}  is
feasible, then there is a feasible solution with no arc carrying flow from
$V\setminus Z$ to $Z$.

Thus we can reduce the problem to solving \eqref{lp2mg}  on $V\setminus Z$. 
Let us add an artifical sink node $t$ to $V$. For every
$i\in V\setminus Z$ with $b_i>0$, 
add a new arc $ti$  with gain factor
$\gamma_{ti}=1$. For every $i\in V\setminus Z$, add an $it$ arc with $\gamma_{it}=1/\bar B$.
This gives an instance of \eqref{primal} with sink
$t$. The condition  \eqref{cond:root} is guaranteed by the $it$ arcs;
for \eqref{cond:init}, we have a simple feasible solution:
send $b_i$
units of flow on $\gamma_{ti}$ for every $i$ with $b_i>0$, and set the
flow to 0 on all other arcs. The boundedness condition
\eqref{cond:bounded} 
is guaranteed by Lemma~\ref{lem:boundedness}; note that by the definition
of $Z$, there are no flow generating cycles in $E[V\setminus Z]$.

\begin{lemma}
Let $f$ be an optimal solution to the  \eqref{primal} instance as
constructed above.
Then \eqref{lp2mg} is feasible if and only if $f_{ti}=0$ for all $i\in
V\setminus Z$.
\end{lemma}
\begin{proof}
Consider an optimal solution $f$ to the instance of \eqref{primal}
with an optimal labeling $\mu$. If $f_{ti}=0$ for 
 $i\in V\setminus Z$, then $f$ restricted to $V\setminus Z$ is a
 feasible solution \eqref{lp2mg}. Conversely, assume $f_{tj}>0$ for a
 certain node $j\in V\setminus Z$; we show that  \eqref{lp2mg} is infeasible.

As in the proof of Lemma~\ref{lem:transform-prop}, we let $S_0$ denote
the set of nodes $i\in V\setminus Z$ with $\mu_i=\bar B$, and let $S$
be the set of nodes that can be reached from $S_0$ on a residual path
in $E_f$. We claim that $j\notin S$. To see this, first observe that
$\mu_j=1$ because of $f_{tj}>0$.
If there were a path $P\subseteq E_f$ from a node $i\in S$
to $j$, then $1\ge \gamma^\mu(P)=\gamma(P)\mu_i/\mu_j=\gamma(P)\bar B$
gives a contradiction to the choice of $\bar B$.

Therefore $X=V\setminus (Z\cup S)$ contains $j$, and there is no arc
entering this set. Further,
 $e_i(f)=0$  and  $f_{it}=0$ for every $i\in X$. Then
 $y_i:=1/\mu_i$ for  $i\in X$ and $y_i:=0$ for $i\notin X$ gives a Farkas certificate
of infeasibility for \eqref{lp2mg}.\footnote{The Farkas
certificate is described after the proof of Lemma~\ref{lem:both-feas}.}
 Indeed, $y\ge 0$,  $y_i-y_j\gamma_{ij}\ge 0$ holds for every arc
$ij\in E$, and $\sum_{i\in V}b_iy_i>0$ because
\[
\sum_{i\in V}b_iy_i=\sum_{i\in X}b_i^\mu=
\sum_{i\in X} \sum_{j\in  V\cup\{t\}:ji\in E} \gamma_{ji}^\mu f_{ji}^\mu-\sum_{j\in
  V:ji\in E} f_{ji}^\mu=\sum_{i\in V} f_{ti}^\mu>0,
\]
 completing the proof. Here we used that $\gamma_{ji}^\mu=1$ whenever $f_{ji}>0$.
\end{proof}

\subsubsection*{Solving  \eqref{lp2m}}
We solve \eqref{lp2mg} as described above, and \eqref{lp2ml} the same
way, after reversing the arcs. If either of the two problems is
infeasible, then \eqref{lp2m} is also infeasible. Assume now that $f$
is a feasible solution to 
\eqref{lp2mg}, and $g$ is a feasible solution to \eqref{lp2ml}. We
show that in this case the equality version \eqref{lp2m} is also
feasible. To prove this, we use
a flow decomposition of the difference of the two solutions $f$ and
$g$ to transform $g$ to a solution of \eqref{lp2ml}. 

\begin{lemma}\label{lem:both-feas} 
Given feasible solutions to \eqref{lp2mg} and \eqref{lp2ml}, a
feasible solution to  \eqref{lp2m} can be found in strongly polynomial time.
\end{lemma}
\begin{proof}
Let  $f$ be a feasible solution to 
\eqref{lp2mg}, and $g$  a feasible solution to \eqref{lp2ml}. Then for
every $i\in V$, $e_i(f)\ge 0\ge e_i(g)$ holds.
Let us define the flow $h$ as
\[
h_{ij}:=
\begin{cases}
f_{ij}-g_{ij} &\mbox{ if }ij\in E,\ f_{ij}> g_{ij}\\
\gamma_{ji}(f_{ji}-g_{ji}) &\mbox{ if }ji\in E,\ f_{ji}>g_{ji}.
\end{cases}
\]
Let $H\subseteq \ole E$ denote the support of $h$; clearly, $h_{ij}>0$
for every $ij\in H$. With the convention $h_{ij}=-\gamma_{ji}h_{ji}$,
we have
$f=g+h$.
Since $e_i(f)\ge 0\ge e_i(g)$, the inequality
$\sum_{j:ji\in  H}\gamma_{ji}h_{ji}\ge \sum_{j:ij\in H}h_{ij}$ holds for
every $i\in V.$

We apply the standard generalized flow decomposition for $h$ as in
e.g.  \cite{Gondran84,Goldberg91}: every generalized flow can be
written as the sum of five types of  elementary flows. Such a decomposition can be
found in $O(nm)$ time, and the number of terms is at
most the number of arcs with positive flow.

Among the five types of elementary
flows listed in \cite{Goldberg91}, Types I and III cannot be present
the decomposition of $h$, as there are no deficit nodes (more outgoing than
incoming flow). Type IV are unit gain cycles, and Type V are pairs of
flow generating and flow absorbing cycles connected by a path (``bicycles''); these
do not generate any excess or deficit and are not needed for out
argument.
The important one is Type II: a flow generating cycle and a path
connecting it to an excess node (more incoming than outgoing flow).

We now describe how to modify $g$ to a feasible solution $g'$ to
\eqref{lp2m} using the decomposition of $h$.
Consider a node $i$ with $e_i(f)>e_i(g)$; this is an excess node for
$h$. We could add all Type II flows in the decomposition terminating at $i$ to increase
$e_i(g)$ to $e_i(f)$. Since we want achieve the equality $e_i(g')=0$,
we only use some of the Type II flows. We add them one-by-one until
$e_i(g')$ becomes nonnegative. Then for the last flow, we add only a
fractional amount to set precisely $e_i(g')=0$.
Repeating this for every $i$ with $e_i(f)>e_i(g)$, we obtain a
feasible solution $g'$ to \eqref{lp2m}.
\end{proof}

We also present a second proof of the claim that
if both \eqref{lp2mg} and \eqref{lp2ml} are feasible, then
\eqref{lp2m} is also feasible. The proof is based on Farkas's lemma
and is not algorithmic, but may contribute to a better understanding
of the claim.

We show that if \eqref{lp2m} is infeasible, then either \eqref{lp2mg}
or \eqref{lp2ml} is also infeasible. A Farkas-certificate to the
infeasibility of \eqref{lp2m} can be written as
\begin{align*}
y_i-y_{j}\gamma_{ij}&\ge 0\quad \forall ij\in E\\
\sum_{i\in V} b_iy_i &> 0
\end{align*}
A Farkas-certificate to the infeasibility \eqref{lp2mg} is the same with the additional
constraint $y\ge 0$, whereas the certificate to the infeasibility of
\eqref{lp2ml} is with $y\le 0$. In the case of \eqref{lp2mg},
$\mu_i=1/y_i$ gives the usual labeling.

Let us define the sets $Y^+:=\{i\in V: y_i>0\}$ and $Y^-:=\{i\in V:
y_i<0\}$. Further, let $y^+_i:=y_i$ if $i\in Y^+$ and $y^+_i:=0$
otherwise; similarly, let $y^-_i:=y_i$ if $i\in Y^-$ and $0$ outside
$Y^-$.

We claim that the $y^+_i-y^+_j\gamma_{ij}\ge 0$ and
$y^-_i-y^-_j\gamma_{ij}\ge 0$ hold for every $ij\in E$.
We only verify this for $y^+$; the proof is the same for $y^-$. If
$i,j\in Y^+$, then this holds because $y^+$ is identical to $y$ inside
$Y^+$. If $i,j\in V\setminus Y^+$, then $y^+_i=y^+_j=0$ and thus the
claim is trivial. Next, let $i\in Y^+$ and $j\in V\setminus Y^+$. The
claim follows by $y^+_i>0$,  and $y^+_j=0$. Finally, we claim that
there is no $ij\in E$ with 
$i\in V\setminus Y^+$, $j\in Y^+$. Indeed, this would
mean $y_i-y_j\gamma_{ij}<0$, contradicting the choice of $y$.

 Since
$0<\sum_{i\in V}b_iy_i=\sum_{i\in V}b_iy^+_i+\sum_{i\in V}b_iy^-_i$, either
$\sum_{i\in V}b_iy^+_i>0$ or $\sum_{i\in V}b_iy^-_i>0$. In the first
case, $y^+$ is an infeasibility certificate for \eqref{lp2mg}, and in
the second case, $y^-$ is an infeasibility certificate for \eqref{lp2ml}.

\section{Conclusion}\label{sec:final}

We have given a strongly polynomial algorithm for the generalized flow
maximization problem, and also for solving feasibility LPs with at
most two nonzero entries in every column of the constraint matrix.
A natural next question is to address the
minimum cost generalized flows, or equivalently, finding optimal
solutions to LPs with two nonzero entries per column.

In contrast to the vast
literature on the flow maximization problem, there is only one weakly
polynomial combinatorial algorithm known for this setting, the one by Wayne \cite{Wayne02}.
This setting is more challenging since the dual structure cannot be
characterized via the convenient relabeling framework, and thereby most tools
for minimum cost circulations, including the scaling approach also
used in this paper, become difficult if not impossible to apply.

\medskip

Another possible line of research would be to extend the flow maximization
algorithm to nonlinear settings. The paper \cite{Vegh11} gave a simple
scaling algorithm for {\em concave generalized flows}, where instead
of the gain factors $\gamma_{e}$, there is a concave increasing
function $\Gamma_{e}(.)$ associated to every arc $e$. In
\cite{Vegh11b}, a strongly polynomial algorithm is given to the
analogous problem of minimum
cost circulations with separable convex cost functions satisfying
certain assumptions. One could combine the techniques of \cite{Vegh11}
and \cite{Vegh11b} with the ideas of the current paper to obtain strongly
polynomial algorithms for some special classes of concave generalized
flow problems. This could also lead to strongly polynomial algorithms for
certain market equilibrium computation problems, see \cite{Vegh11}.

\subsection*{Acknowledgment}

The author is grateful to Joseph Cheriyan, Ian Post, and the anonymous
referees for several suggestions that helped to improve the presentation.

\bibliographystyle{abbrv}
\bibliography{strongly}

\section*{Appendix}\label{sec:intermed}

{
\renewcommand{\thetheorem}{\ref{thm:close-opt}}
\begin{theorem}
Let $(f,\mu)$ be a $\Delta$-feasible pair. Then there exists an optimal solution
$f^*$ such that 
\[
||f^\mu-{f^*}^\mu||_\infty \le Ex^\mu(f)+(|F^\mu|+1)\Delta.
\]
\end{theorem}
\addtocounter{theorem}{-1}
}

\begin{proof}
First, let us modify $(f,\mu)$ to a conservative pair $(\tilde f,\mu)$
by setting the flow values on non-tight arcs to 0, as in
Lemma~\ref{lem:make-conservative}. We shall prove the existence of an
optimal $f^*$ such that
\begin{equation}
||\tilde f^\mu-{f^*}^\mu||_\infty \le Ex^\mu(\tilde f).\label{close-opt}
\end{equation}
This implies the claim, since Lemma~\ref{lem:make-conservative} asserts 
$Ex^\mu(\tilde f)\le Ex^\mu(f)+|F^\mu_f|\Delta$, and $||\tilde f^\mu-f^\mu||_\infty\le \Delta$ as the two flows differ only on non-tight arcs.

Let us pick an optimal solution $f^*$ to (\ref{primal}) such that
$||\tilde f-f^*||_1$ is minimal, and let $\mu^*$ be an optimal
solution to (\ref{dual}). Note that because of (\ref{cond:root}),
all values of $\mu$ and $\mu^*$ are finite.
 We use a similar argument as in the proof of
 Lemma~\ref{lem:both-feas}. Let us define
\[
h_{ij}:=
\begin{cases}
f^*_{ij}-\tf_{ij} &\mbox{ if }ij\in E,\ f^*_{ij}> \tf_{ij}\\
\gamma_{ji}(f^*_{ji}-\tf_{ji}) &\mbox{ if }ji\in E,\ f^*_{ji}>\tf_{ji}.
\end{cases}
\]
Let $H\subseteq \ole E$ denote the support of $h$; clearly, $h>0$
and $H\subseteq E_\tf$ whereas $\obe H\subseteq E_{f^*}$. Again, with the
convention $h_{ij}=-\gamma_{ji}h_{ji}$, we have $f^*=\tf+h$. 

\begin{claim}\label{cl:nocycle}
The arc set $H$ does not contain any directed cycles.
\end{claim}
\begin{proof}
First, let $C\subseteq H$ be a cycle. Since $\mu$ is a conservative labeling for $\tf$ and $C\subseteq E_\tf$, we have $\gamma(C)=\gamma^\mu(C)\le 1$. On the other hand, $\mu^*$ is conservative for $f^*$ and $\obe C\subseteq E_{f^*}$. Therefore 
$\gamma(\obe C)=1/\gamma(C)=1/\gamma^{\mu^*}(C)\le 1$. These together
give $\gamma(C)=\gamma(\obe C)=1$, and also $\gamma^{\mu^*}_e=1$ for
every $e\in C$.
Hence we can modify $f^*$ to another optimal solution by decreasing every ${f^*_e}^{\mu^*}$ value
by a small $\varepsilon>0$. This gives a contradiction to our extremal choice of $f^*$ as the optimal solution minimizing $||\tilde f-f^*||_1$.
\end{proof}

Observe that
\[
e_i(\tf)-e_i(f^*)=\sum_{j:ij\in H}h_{ij}-\sum_{j:ji\in H} \gamma_{ji}h_{ji}
\]
By the optimality of $f^*$, the left hand side is $\le 0$ for
$i=\sink$ and is equal to $e_i(\tf)\ge 0$ otherwise.
The above claim guarantees that $H$, the support of $h$, is acyclic. Consequently, we can easily decompose $h$ to the form
\[
h=\sum_{1\le \ell \le k} h^\ell,
\]
where each $h^\ell$ is a path flow with support $P^\ell$ from a node
$p^\ell$ with $e_{p^\ell}(\tf)>0$ to $\sink$, and $k\le m$.

Such a decomposition is easy to construct by using a topological order
of the nodes for $H$. It is also a special case of the flow
decomposition argument used in Lemma~\ref{lem:both-feas}, see also
\cite{Gondran84,Goldberg91}. 
(The difference is that according to Claim~\ref{cl:nocycle}, four out
of the five types of elementary flows, Types II-V cannot exist as they
contain cycles.)

Let $\lambda^\ell$ denote the value of $h^\ell$ on the first arc of $P^\ell$. 
Since $\mu$ is a conservative labeling and $P^\ell\subseteq H\subseteq E_\tf$, we have $\gamma_{ij}^\mu\le 1$ for all arcs of $P^\ell$ and therefore the relabeled flow $({h^\ell})^\mu$ is monotone decreasing along $P^\ell$. Hence it follows that for every arc $ij$,
\[
h^\mu_{ij}= \sum_{1\le\ell\le k} ({h^\ell_{ij}})^\mu\le  \sum_{1\le\ell\le k}
\frac{\lambda^\ell}{\mu_{p^\ell}}=\sum_{i: V-\sink}e^\mu_i(\tf)=Ex^\mu(\tf).
\]
This completes the proof, since $||\tilde
f^\mu-{f^*}^\mu||_\infty=\max_{ij\in E} h^\mu_{ij}$ (note that if
$f_{ij}^*<f_{ij}$, then $\gamma_{ij}^\mu=1$ must hold).
\end{proof}
\end{document}